\documentclass[12pt]{article}
\usepackage{url}

\usepackage{booktabs} % For formal tables
\usepackage{amsthm, graphicx,times,psfrag,xspace,latexsym,amsmath,amssymb}
\usepackage{algorithm2e}

\newtheorem{claim}{Claim}

\begin{document}
\title{Heterogeneous MacroTasking (HeMT) for \\
Parallel Processing in the Public Cloud}

\author{
\begin{tabular}{cc}
Y. Shan, G. Kesidis, B. Urgaonkar & J. Schad\\
Pennsylvania State University & Mesosphere \\
University Park, PA & San Francisco, CA \\
\{yxs182,gik2,buu1\}@psu.edu &  joerg@mesosphere.io \\
 & \\
 J. Khamse-Ashari, I. Lambadaris & \\
 Carleton University & \\
 Ottawa, Canada & \\
\{jalalkhamseashari, ioannis\}@sce.carleton.ca & 
\end{tabular}
}
\maketitle

\noindent {\bf Abstract - }
Using tiny, equal-sized tasks (Homogeneous microTasking, HomT) has long been
regarded an effective way of load balancing in parallel computing systems.
When combined with nodes pulling in work upon becoming idle, HomT has the desirable
property of automatically adapting its load distribution to the processing capacities of
participating nodes - more powerful nodes finish their work sooner and, therefore, pull in
additional work faster. As a result, HomT is deemed especially desirable in settings with
heterogeneous (and possibly possessing dynamically changing) processing capacities.
However, HomT does have additional scheduling and I/O overheads 
that might make this load balancing scheme costly in some scenarios. In this paper,
we first analyze these advantages and disadvantages of HomT. We then propose an alternative
load balancing scheme -  Heterogeneous MacroTasking (HeMT), wherein workload is
{\em intentionally} partitioned according to nodes' processing capacity. 
Our goal is to study when HeMT is able to overcome the performance disadvantages of HomT.
We implement a prototype of HeMT within the Apache Spark application framework with
complementary enhancements to the Apache Mesos cluster manager. Spark's built-in scheduler,
when parameterized appropriately, implements HomT. Our experimental results show that HeMT
outperforms HomT when accurate workload-specific estimates of nodes' processing capacities
are learned. As representative results, Spark with HeMT offers 
about 10\% better average
completion times for realistic data processing workloads over the default system.

\section{Introduction}

Parallel data processing represents a large and important class of workloads running on
public cloud computing platforms such as Amazon Elastic Computing Cloud (EC2), Google Compute
Engine (GCE), and Microsoft Azure. Load balancing - dividing work among the nodes 
of a cluster -  has an important role in determining the performance of these workloads.
The clusters underlying such workloads are often prone to exhibiting {\em heterogeneity} in
the processing capacities of their constituent nodes. Such heterogeneity may have
negative implications for workload performance (such as completion times) and cost incurred
by the workload owner or ``tenant'' (due to resource wastage in under-utilized nodes). 

Heterogeneity may arise from myriad causes including (i) an intentional procurement of such nodes
by the tenant (mandated by some other requirements) or (ii) time-varying resource interference
within the cloud infrastructure. Budget-conscious tenants (such as academic researchers
operating on limited budgets but also fledgling startups operating in a cut-throat market)
are especially likely to contend with heterogeneity due to (ii). 
Such tenants often procure relatively cheap 
virtual machines (VMs) such as spot or  burstable or small-sized regular instances. These
``wimpy'' instances are created by providers to monetize their dynamic spare capacity and are
managed using aggressive statistical multiplexing/overbooking of physical resources. As a
result, such instances are known to exhibit significant dynamism and unpredictability in
their effective capacities~\cite{sigmetrics17,EuroSys17}.

{\bf The Problem:} These heterogeneous settings  exacerbate the well-known
``straggler problem''~\cite{mantri}, wherein a small subset of slow tasks stall an entire
parallel computation that they are part of by causing a
synchronization delay at a program barrier (i.e., {\em all} parallel tasks need to  complete
before the program can proceed and end up waiting for the slow task(s)).
While we offer a broad survey of salient solutions/ideas Section~\ref{sec:related}, our
specific interest in this paper lies in the following: {\em what are the relative pros and
  cons of approaches based on fine-grained vs. coarse-grained partitioning
  in overcoming heterogeneity-induced slowdowns?}

As a prominent representative of the
former approach, consider \cite{Ousterhout13} which advocates that parallel jobs should be
divided into relatively small-sized tasks (``microtasks'') via fine
homogeneous partitioning of the input dataset on which processing is being
performed\footnote{Specifically, \cite{Ousterhout13} suggests microtasks take on the order
  of $100$ ms to execute on contemporary systems.}. Microtasking can lead to good
load balancing when combined with a ``pull-based'' operation: when underbooked or idle,
nodes pull work  (tasks) from a pending queue so faster workers simply pull in more work. 
Furthermore, the synchronization delays are reduced without needing knowledge of either
the speed of the nodes %(task runners)
or the  resources required to achieve particular execution-times
for the tasks. Because of this property, we refer to such approaches as {\em oblivious}
load balancing. However, there also exist studies, e.g., \cite{against-tiny}, that
challenge the microtasking idea, pointing out that the relatively large overhead of
microtasking can, in some cases,  significantly slow down computation. 

\noindent {\bf Research Approach and Contributions: } To explore the microtasking vs. macrotasking trade-off,  we 
  leverage initial modifications we have made to Apache Spark~\cite{spark-the-paper},  a popular
  parallel data processing application framework. 
  Our modifications have been designed with the goal of  
  enabling a more intelligent, cost-conscious tenant to use cloud resources more 
efficiently.
Being cost-conscious, such a tenant may have selected instances
that just meet its requirements (rather than more expensive instances whose capacity
is likely to go idle occasionally). This may involve
custom VMs/containers as in Resources as a Service (RaaS) \cite{raas}, 
less expensive spot/revocable instances, 
burstable/bursting instances (with only intermittently available resources),
or ``wimpy"  regular instances 
that are allocated, e.g., only a fraction of a core.
Furthermore, such a tenant may have characterized its
workload's resource needs to achieve certain performance goals (``demand-side'' characterization).  
Finally,  such a tenant may also employ suitable prediction and monitoring techniques - we
present some examples in 
Sec. \ref{sec:adapted-hemt} and \ref{sec:HeMT}
 - that let it estimate the (relative)
effective capacities of its cluster's nodes in the near future (``supply-side''
characterization).  Using homogeneous microtasking
as a reference, we investigate how/when the straggler problem can be mitigated
through heterogeneous partitioning using such demand- and supply-side characterization. 

The contribution of this paper is fourfold.
\begin{itemize}
\item
  Although microtasking can provide certain qualities-of-service  and efficiencies without detailed information about the cluster or workload, its usefulness may be hindered by its overhead. We interpret part of this overhead via a simple analytical model. 
\item
  We consider variants of heterogeneous macrotasking (HeMT) corresponding to different degrees of accuracy/certainty in
  supply/demand characterization ranging from an oblivious, incrementally-adjusted HeMT to
  a more sophisticated version where offline/online knowledge of node capacities 
is also leveraged.
\item
  Using a variety of experiments on our Spark/Mesos prototype on Amazon EC2 with different
  workloads and nodes (ranging from regular to burstable EC2 instances), we show the
  efficacy of our ideas. Our final set of experiments employ two important multi-stage
  workloads, i.e., PageRank and K-Means.
\item
  We identify and suggest several interesting directions for future work, especially
  related to adaptive scheduling across the application frameworks and cluster manager layers, including improved information exchange between them via enhanced APIs. 

\end{itemize}

\noindent {\bf Outline:} This rest of this paper is organized as follows. 
We overview the Spark application framework (and Mesos
cluster manager) in Sec. \ref{sec:bg}.
In Sec. \ref{sec:homt}, we discuss (homogeneous) microtasking in detail.
Heterogeneous macrotasking is introduced in Sec. \ref{sec:bg-hemt}.
In Sec. \ref{sec:adapted-hemt}, we study a simple
``oblivious" approach to online adapting the size of heterogeneous
macrotasks based on synchronization delays (variations
in task execution times) at program barriers.
Heterogeneous macrotasking for workers based on  statically
provisioned containers or on burstable instances is studied in 
Sec. \ref{sec:HeMT}.
Multistage workloads are considered in Sec. \ref{sec:multistage}.
Other related work is discussed in Sec. \ref{sec:related}. 
After a discussion of 
future work in  Sec. \ref{sec:fw}, we
conclude in Sec.  \ref{sec:conclus}.

\section{Overview of the Application Framework Spark and 
the Cluster Manager Mesos}\label{sec:bg}

\begin{figure}[ht!]
	\centering
        \includegraphics[width=0.8\columnwidth]{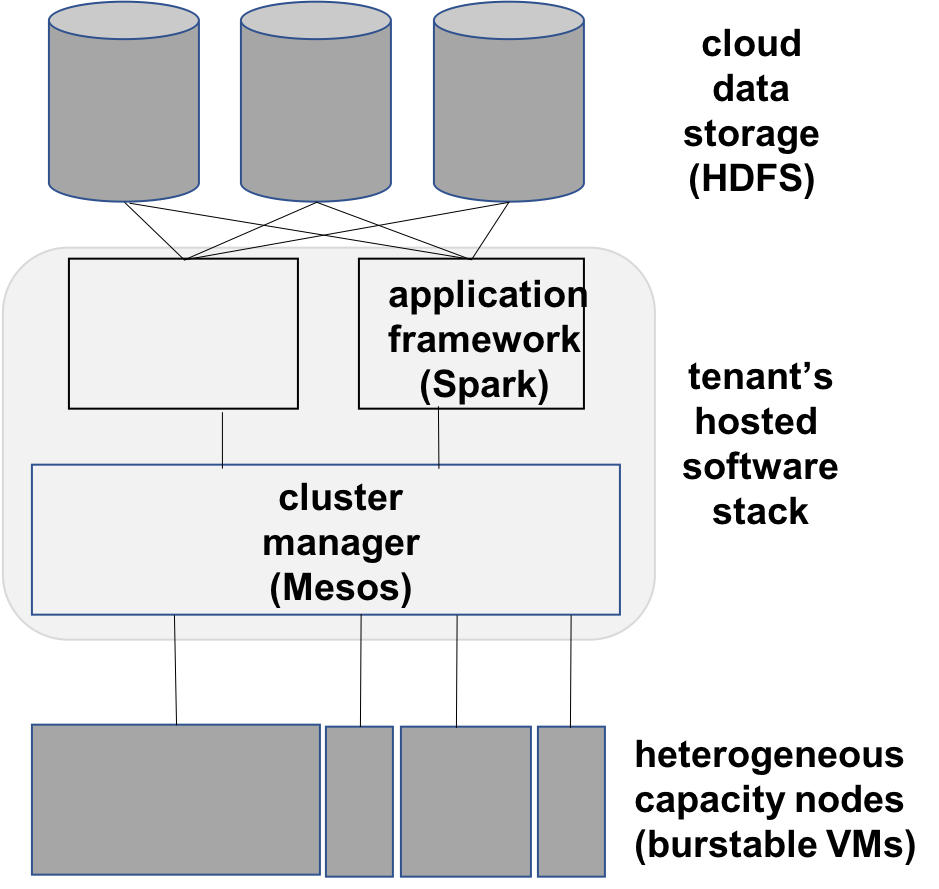}
	\caption{Overview of a cloud tenant under consideration hosting multiple 
          application frameworks and a cluster manager. Application frameworks run their computations using the virtual machine resources allocated to them by the tenant's cluster manager and use a cloud-based facility for persisting their data. The virtual machines procured from the cloud may be heterogeneous, setting up the challenges investigated in this paper.  }
	\label{fig:overall-arch}
\end{figure}

The architecture of a  typical cloud tenant is shown in Fig. \ref{fig:overall-arch}. The tenant's users first register their jobs (application frameworks) with the cluster manager. The cluster manager will then allocate resources from the managed (virtual or physical machines) to the users' jobs so that they can run in a distributed fashion.
In this paper, the example cluster manager is based on Apache Mesos
and the distributed application frameworks are based on
Apache Spark.

Apache Mesos \cite{mesos} is a widely-used cluster manager. To start a Mesos cluster, a Mesos master process has to be started and Mesos agent processes, responsible for reporting available resources to the master, need to run on each resource-providing machine (often referred as ``computing node'' or simply ``node'' in the following) and register with the master. The Mesos master informs the registered frameworks about resource availabilities through resource offers. Upon receiving such an offer, a framework decides how much resources it will use and informs the Mesos master.

Apache Spark \cite{spark-the-paper} is a distributed computation framework which we run over Mesos. The Spark driver, running on either cluster's machine or user's own machine, acts as a centralized job orchestrator. It divides a job into multiple computation stages. The mutually dependent stages are separated by data shuffling. A stage can only start when all its dependent stages have been completed. Each stage may contain multiple parallel tasks. To run those tasks, the Spark driver first acquires resources from its cluster (via Mesos) to launch executors (task runners) on the distributed nodes. Without considering any task placement constraints such as locality preferences or blacklists, the Spark driver sends the pending tasks to the idling executors.

Spark's parallelism, i.e., how many parallel tasks can run on a single stage, by default depends on the number of computing slots (e.g., the number of CPU cores on those spawned executors) or the distribution of the input data (e.g., how many blocks in the input data file). Users can configure their preferred parallelism if they want more tasks, each processing less workload, or fewer tasks, each processing more workload. 
Spark tasks are sized equally or according to data locality, without considering the processing speed of executors.

\section{Homogeneous Microtasking (HomT)}\label{sec:homt}

Homogeneous microtasking, i.e., equally dividing a job into tiny tasks which greatly out-number the number of computing slots, can well balance the workload without accurately knowing the computation speed of the computing nodes in the cluster 
\cite{Ousterhout13, lb-mapreduce}. This is supported
by the upper-bound on resource idling time stated in the following claim.

\begin{claim}
	For pull-based task assignment (i.e., an executor pulls one more task if it is done with its assigned task and there are more tasks remaining), suppose all the tasks in a stage are pending at some initial time 0, the workload is evenly partitioned into tasks and the speed of the nodes is constant. Then the resource idling time (latest node finish time minus earliest node finish time) is upper-bounded by the single task duration of the slowest node.
\end{claim}
\begin{proof}
	Let $T_i$ be the task duration on node $i$, $t_i$ be the finish time of node $i$ (the time when node $i$ finishes its last task). Assume node $0$ is the first node to be idle. Note that $\forall \epsilon > 0$, at $t_0 + \epsilon$, there is no pending task\footnote{Otherwise, node $0$ would pull another pending task to run.} and all the remaining tasks are running on the other nodes. Let $t_0 - \Delta t_i, i \neq 0$ be the starting time of the running task on node $i$, so the finish time of the last task would be $\max_i\{t_0 -\Delta t_i + T_i\} = t_0 + \max_i\{T_i - \Delta t_i\}$. Note that $\Delta t_i \geq 0, \forall i, i \neq 0$. Therefore,  $t_0 + \max_i\{T_i - \Delta t_i\} \leq t_0 + \max_i\{T_i\}$.
\end{proof}

However, the advantages of homogeneous microtasking are limited under several scenarios.

First, the scheduling overhead grows with the number of tasks. This problem makes microtasking approach less practical in some distributed computing frameworks such as Apache Hadoop whose tasks could take up to several seconds to launch \cite{Ousterhout13}.

Second, as discussed in previous work \cite{riffle, Ousterhout13, flat-dc}, one of the major concerns of dividing a job into tiny tasks is disk I/O inefficiency. When a large task is further divided into many tiny ones, one original sequential disk read would be divided into many small random disk reads, hence leading to relatively lower I/O throughput. Also, if a task input is small, then depending on the reading buffer (normally ranging from kBs to MBs), the task may finish after only a couple of I/O requests, so the advantage of pipelined read-process may no longer exist. So increase in stage completion times result,  
as shown in
Figs. \ref{container04}
and \ref{100bw}-\ref{40bw} 
in Sec. \ref{sec:containers} and \ref{sec:burstables}, respectively.

Another problem of microtasking we found in our experiments 
is related to the distributed file system, e.g., 
the Hadoop Distributed File System (HDFS)
\cite{hdfs-design}: when the job is network I/O bottlenecked, if multiple tasks simultaneously access the same HDFS data block, then they are more likely to read on the same datanode, which may lead to inefficient overall CPU and network bandwidth usage in the HDFS cluster. On the other hand, Spark sequentially schedules tasks so that
consecutive tasks are more likely to access the same block as large tasks are 
increasingly divided into smaller ones.

HDFS is designed to store very large files in a distributed fashion. It follows a master/slave architecture. Namenode, the master, manages the file system metadata and coordinates the data placement. Datanodes serve as slaves that  perform the actual data reads and writes. To operate on HDFS, a client first contacts the namenode, which will redirect the client to the correct datanodes for write and read. The HDFS architecture is illustrated in Fig. \ref{hdfs-arch}.

 \begin{figure}[ht!]
	\centering
	\includegraphics[scale=0.28]{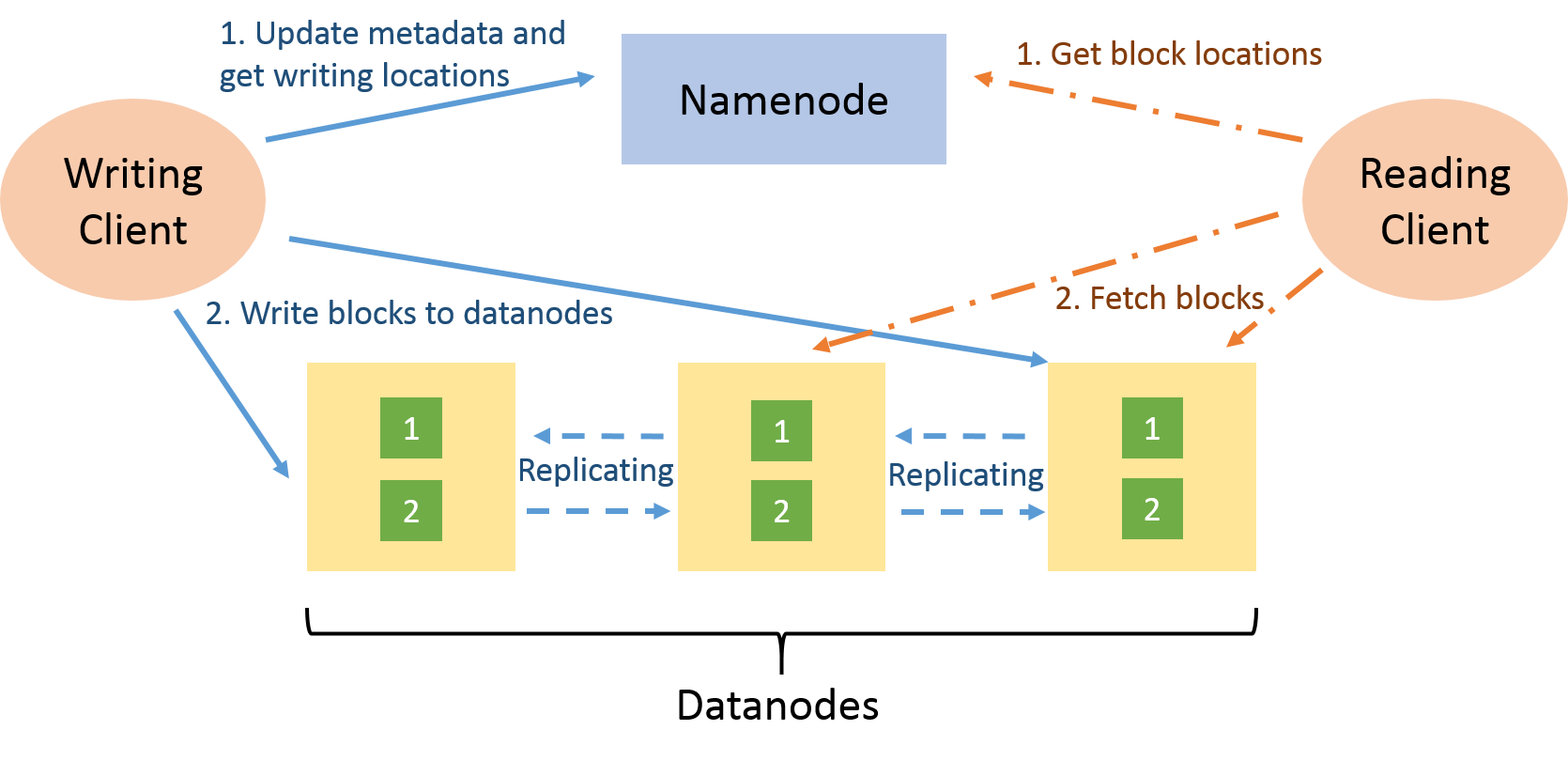}
	\caption{HDFS architecture.}
	\label{hdfs-arch}
\end{figure}

Each file is usually divided into a sequence of blocks, and each block is replicated for fault tolerance. HDFS does not allow a single datanode to store multiple replicas for the same block \cite{hdfs-design}. 
To simplify our analysis, we make two assumptions: First, rack-awareness \cite{hdfs-rack-awareness} in block placement is turned off\footnote{In fact, Hadoop rack-awareness has  less randomness and thus intensifies uplink competition
since data blocks are less broadly spread.}. 
Second, a simple placement policy is assumed such that 
the namenode randomly chooses one among the equally-distant datanodes to place a data block as well as its replicas. So when a remote user uploads a data block, HDFS randomly chooses datanodes, each storing exactly one replica of that block. Upon a read request, HDFS tries to choose the closest replica (on the same datanode or on the same rack) to the reading client. If there are multiple candidate replicas, then the replica choice is distributed 
uniformly at random among them.
One typical HDFS replica distribution is shown in Fig. \ref{hdfs-replica-distr}.

\begin{figure}[ht!]
	\centering
	\includegraphics[scale=0.28]{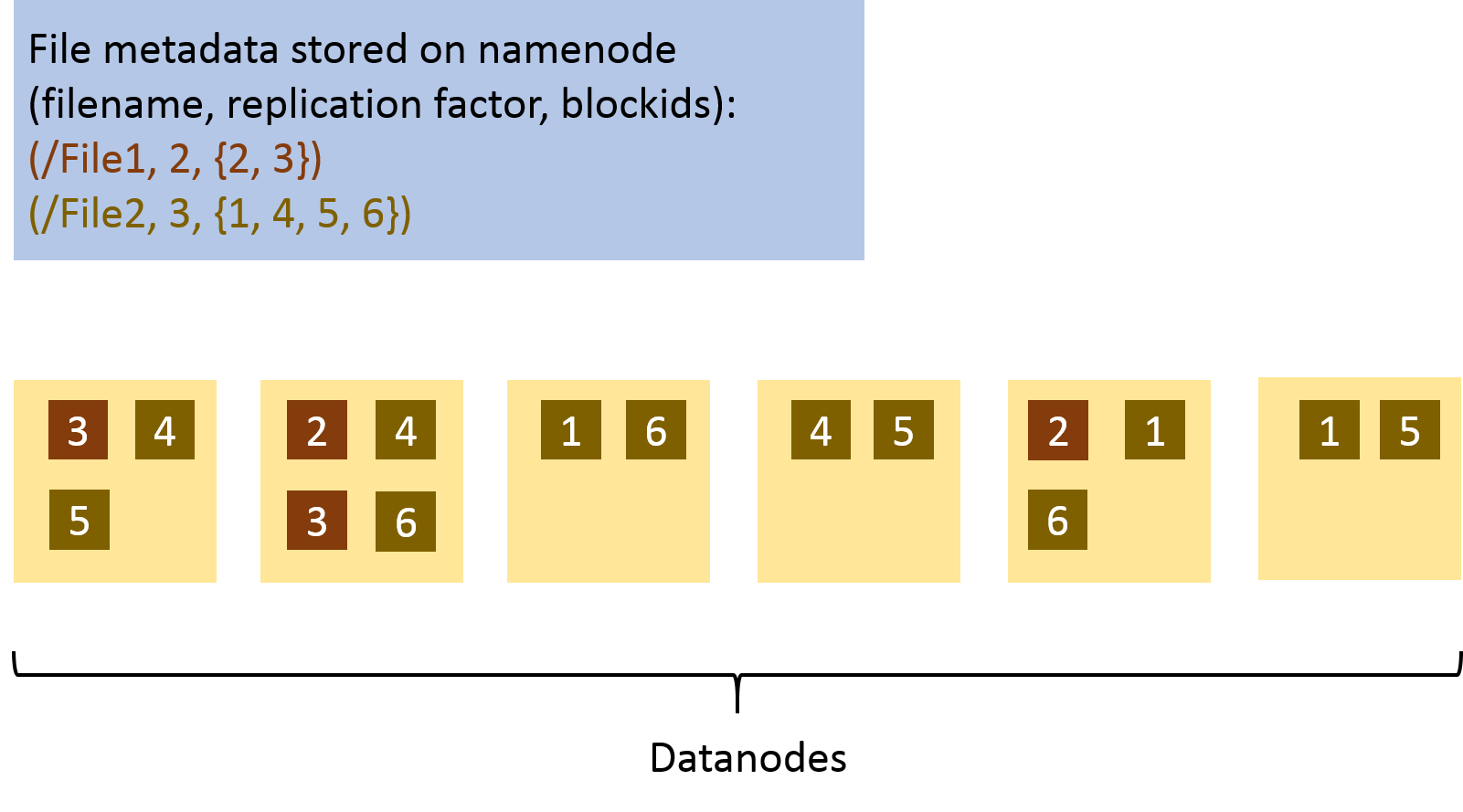}
	\caption{An illustrative example of HDFS replica distribution.}
	\label{hdfs-replica-distr}
\end{figure}

In a HDFS cluster, suppose the number of datanodes is $n$ and the replica factor is $r$, and $n \geq r$ (the usual case). If two tasks access the same HDFS block, then the probability that they will read on the same datanode, competing for its uplink bandwidth\footnote{They may also compete for the disk bandwidth on that datanode. But considering disk bandwidth is usually larger than the network bandwidth, disk bandwidth is not the concern herein.} is
\begin{equation}\label{expr_p1}
p_1 = 1/r.
\end{equation}

If two tasks access different HDFS blocks, the probability that they will read on the same datanode is
\begin{equation}\label{expr_p2}
p_2 = \sum_{v = \max(2r - n, 0)}^{r} \textrm{P}(v) \frac{v}{r^2},
%GK: summand is (v/r)*(1/r)P(v)
\end{equation}
where $\textrm{P}(v)$ is the probability that there are $v$ datanodes that store replicas of both blocks. Assume HDFS randomly choose $r$ datanodes to store replicas when the data is uploaded (appears so, need to verify if it is true in HDFS implementation),
\begin{equation}\label{expr_pv}
\textrm{P}(v) = \frac{\binom{r}{v} \binom{n - r}{r - v}}{\binom{n}{r}}.
%GK: equivalently, P(v) = 
%\binom{n}{v}\binom{n-v}{r-v}\binom{n-r}{r-v}/\binom{n}{r}^2
\end{equation}

Considering Eq. (\ref{expr_pv}) and (\ref{expr_p1}), we can

\begin{claim}
\begin{eqnarray*}
 p_1 & \geq &  p_2,
\end{eqnarray*}
with equality when $r = n$.
\end{claim}

\begin{proof}
	With Eq. \ref{expr_pv},
	\begin{eqnarray}
	p_1 \geq p_2 &\Leftrightarrow& \sum_{v = \max(2r - n, 0)}^{r} \frac{\binom{r}{v} \binom{n - r}{r - v}}{\binom{n}{r}} \frac{v}{r^2} \leq \frac{1}{r} \nonumber\\
	&\Leftrightarrow& \sum_{v = \max(2r - n, 0)}^{r} \frac{\binom{r}{v} \binom{n - r}{r - v}}{\binom{n}{r}} \frac{v}{r} \leq 1 \label{p1p2_prov}
	\end{eqnarray}
	The inequality in (\ref{p1p2_prov}) obviously holds by noting that within the summation in (\ref{p1p2_prov}), $\frac{v}{r} \leq 1$, and
	\begin{equation*}
	\sum_{v = \max(2r - n, 0)}^{r} \frac{\binom{r}{v} \binom{n - r}{r - v}}{\binom{n}{r}} = 1.
	\end{equation*}
\end{proof}

The plot of $p_1$ and $p_2$ versus $n$ for $r = 2$ 
%or $3$ (HDFS default),
shown in Fig. \ref{p1p2_2dn}
%and Fig. \ref{p1p2_3dn} respectively, 
numerically supports the above conclusion that $p_1 \geq p_2$. That is, they indicate that two tasks reading on the same block are more likely to compete for the uplink bandwidth on the same datanode.

\begin{figure}[ht!]
	\centering
	\includegraphics[scale=0.6]{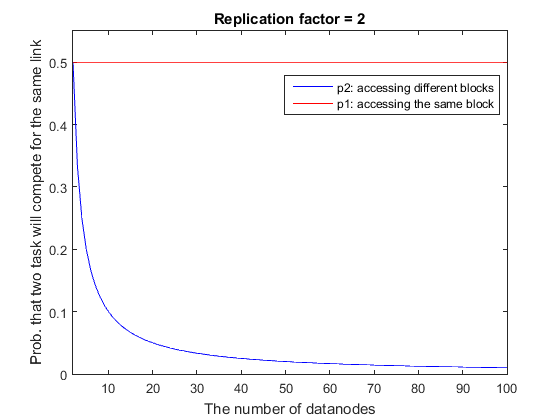}
	\caption{$p_1, p_2$ vs. $n$ when the replication factor is $2$.}
	\label{p1p2_2dn}
\end{figure}

%\begin{figure}[ht!]
%	\centering
%	\includegraphics[scale=0.5]{hdfs_conflict_prob_3datanodes.png}
%	\caption{$p_1, p_2$ vs. $n$ when the replication factor is $3$.}
%	\label{p1p2_3dn}
%\end{figure}

In another experiment to support the above analysis, a small HDFS cluster, with $n=4$ datanodes and replication factor $r = 2$, had limited datanode uplink-bandwidth of $64$ Mbps. Thus, the Spark tasks are always bottlenecked by network I/O.
The datanodes are equally distant to Spark's computing nodes so the selection of datanode for one block is random. The results are shown in Fig. \ref{10bw}. As can be observed, the stage completion time increases with the number of
tasks/partitions.

\begin{figure}[ht!]
	\centering
	\includegraphics[scale=0.38, trim={2cm, 0, 0, 0}]{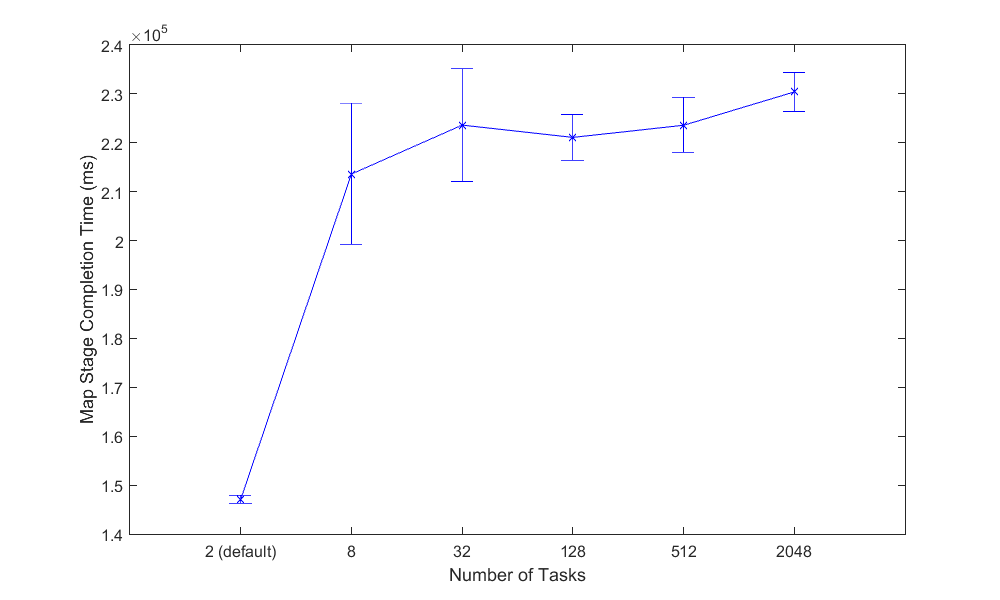}
	\caption{The stage completion time with different partitioning granularity, when network I/O is the universal bottleneck.}
	\label{10bw}
\end{figure}

\section{Heterogeneous Macrotasking (HeMT) - Background}\label{sec:bg-hemt}

To avoid HomT overhead, the number of tasks can be set equal to the
number of available ``computation slots'' (executors). However, in 
case of heterogeneous executors, synchronization delay may ensue
if such ``macrotasks'' are equally sized. This motivates heterogeneous
macrotasking (HeMT).

HeMT will  require a
reasonably accurate estimation of workload (reflected by task execution time) 
which can be easily obtained for many modern jobs due to their repetitive nature;
e.g.,  many production workloads \cite{workload-enterprise-dc, workload-windows-servers, fb-data-warehouse}, and many machine-learning related jobs such EM  and K-Means \cite{duda-hart-stork-01} that consist of multiple iterations of the same computational complexity. Much recent work on task scheduling, e.g., \cite{tr-spark}, is based on such an assumption.
	
We implemented this HeMT partitioning algorithm on 
Spark and compared
 it with Spark's default partitioning scheme in the following,
as well as the aforementioned HomT.
 Spark's default partitioning does not consider any resource heterogeneity of the cluster - it divides the input data regardless of the speed of computing nodes - and 
Spark tends to evenly divide on-memory data into as many partitions as the number of computing slots (usually processing cores). For files located on disk, e.g. HDFS files, baseline Spark, like Hadoop, assigns one file block to a task. Spark naturally supports HomT: users can specify a desired number of partitions and Spark would evenly divide data according to this number.

The aim the of experimental studies described in the following
is to illustrate the benefits and challenges of HeMT.
We implemented HeMT in Spark \cite{spark-code}  using
information from middleware 
(here, Apache Mesos cluster manager \cite{mesos-code}) 
or directly from monitoring services (e.g., AWS CloudWatch).
For scalability, the application frameworks perform most elements of
(workload specific)
HeMT learning, while the middleware scheduler may only perform more
sophisticated workload scheduling (consolidation)\footnote{Analogies 
can be made with ``end-to-end" approaches
such as exokernels or TCP congestion control in the Internet.}.
The information exchange in our Spark-Mesos prototype is summarized in 
Figure \ref{fig:APIs}.

\begin{figure}[ht!]
	\centering
        \includegraphics[width=\columnwidth]{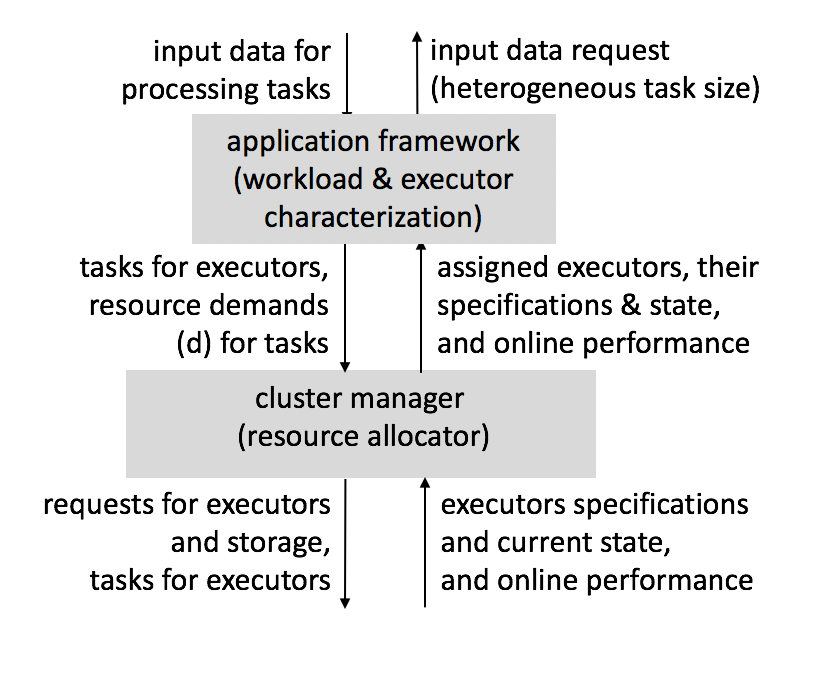}
        \caption{Overview of proposed modifications to application
frameworks and cluster  manager.}
	\label{fig:APIs}
\end{figure}

\section{Oblivious Adapted HeMT (OA-HeMT)}\label{sec:adapted-hemt}

In some environments, e.g., those without resource isolation
leading to significant interprocess 
interference, determining the true workload processing power 
of available computational nodes may be challenging. 
So, a simple ``oblivious'' approach is needed to allow 
application frameworks and cluster managers to dynamically estimate the processing speed of available computational nodes according to the previous workloads of the same job, so that the future tasking 
can be well balanced.

A Spark-Mesos prototype was implemented to enable such an oblivious
ad-hoc adaptive HeMT. The  Mesos cluster manager 
obtains and passes on to the Spark application framework
estimated executor processing speed
through additional fields in their RPC messaging.
Based on this information and the associated task sizes, 
Spark estimates the execution speed of different available
executors and thereby determines how to partition future
work into well-balanced tasks. 

\subsection{A general approach}

Consider a sequence of datasets of sizes $\{D_k\}$ that need to be
processed in the same way, i.e., the same job applied
to each dataset.
The $k^{\rm th}$ dataset $D_k$ is divided
(by the application framework)
into a number of tasks, one for each executor 
$i\in L_k$ assigned to process the $k^{\rm th}$ dataset $D_k$
(by the cluster manager). These tasks are created by dividing
the dataset $D_k$.

For each executor $i\in L_k$, let $v_i$ be the most recent
estimate of its ``speed" for the job under
consideration. Let $L^o_k\subset L_k$ be the set of executors
that have not before been assigned to this job.
For all $i\in L^o_k$, let $v_i = \overline{v}$ where
$\overline{v}$ is the average $v_j$ for $j\in L_k\backslash L^o_k$
(example other choices could be the minimum or maximum
rather than the average or the average speed over all
executors that have been applied to this job in the past).
Let
\begin{eqnarray*}
V_k & = & \sum_{i\in L_k} v_i 
%~ = ~ |L_k^o|\overline{v}+\sum_{i\in L_k\backslash L^o_k}v_k
 ~ = ~ |L_k|\overline{v},
\end{eqnarray*}
where $|L_k|$ is the number of executors assigned to the $k^{\rm th}$ job.
Executor $i\in L_k$ is assigned a dataset of size $d_i = D_k v_i/V_k$.
That is, the faster executor (larger $v_i$)
is assigned to work on a larger dataset (larger $d$).

Let $t_i$ be the execution time of executor $i\in L_k$ on 
the assigned task of size $d_i$ of the $k^{\rm th}$ job.
For all executors $i\in L_k$, their 
speed can be updated according to a simple 
first-order autoregressive estimator
\begin{eqnarray*}
v_i & \leftarrow & (1-\alpha) \frac{d_i}{t_i} + \alpha v_i
\end{eqnarray*}
where forgetting factor $\alpha$ satisfies $0<\alpha<1$.

For the initial  ($k=1$) job, $D_1$ is evenly divided
among the executors $i\in L_1$ and subsequently
$v_i = d_i/t_i$.

The straightforward tradeoff in the choice of $\alpha$ is
that smaller $\alpha$ means that the speed estimate is more
responsive to the latest speed datapoint $d_i/t_i$.
But it's entirely possible that different datasets 
of the {\em same size}, i.e., $d=d'$, will require
different execution times $t\not=t'$ for the
same job type under consideration. Over time, such
variations will be ``averaged out" in the 
executor speed estimates; i.e., each executor will 
experience the same task-difficulty distribution 
``per unit" input data (unless there is some bias
so that  some executors tend to receive more difficult tasks per unit input
data for a given job). This motivates
a forgetting factor $\alpha$ that this not close to zero.

Note that each application framework (different job types)
will need to maintain its own estimates of (workload specific)
executor speeds.

\subsection{An experimental result}

To see the effect of such adaptive workload partitioning, 
we performed an experiment
 with a two-node cluster where each node provides one CPU core. No resource isolation technology was used, so Spark executors could share CPU cycles with other processes. A sequence of fifty Spark WordCount jobs 
were presented through a submission queue. 
We introduced interfering processes \cite{sysbench} on one node at 
two different points in time during the experiment  thus reducing the processing speed of Spark executors on that node. How Spark jobs were adaptively partition to re-balance their workloads is shown in Fig. \ref{workload-rebalance}.

One can see how 
overall job execution times (determined by the slowest task)
increased dramatically but then 
rapidly fell as the task sizes were adapted with 
zero forgetting factor
(here, for a given executor, execution time variation per unit document size -
measured in MBytes - was low). 

\begin{figure}[ht!]
	\centering
	\includegraphics[scale=0.45, trim={0, 0.5cm, 0, 0}]{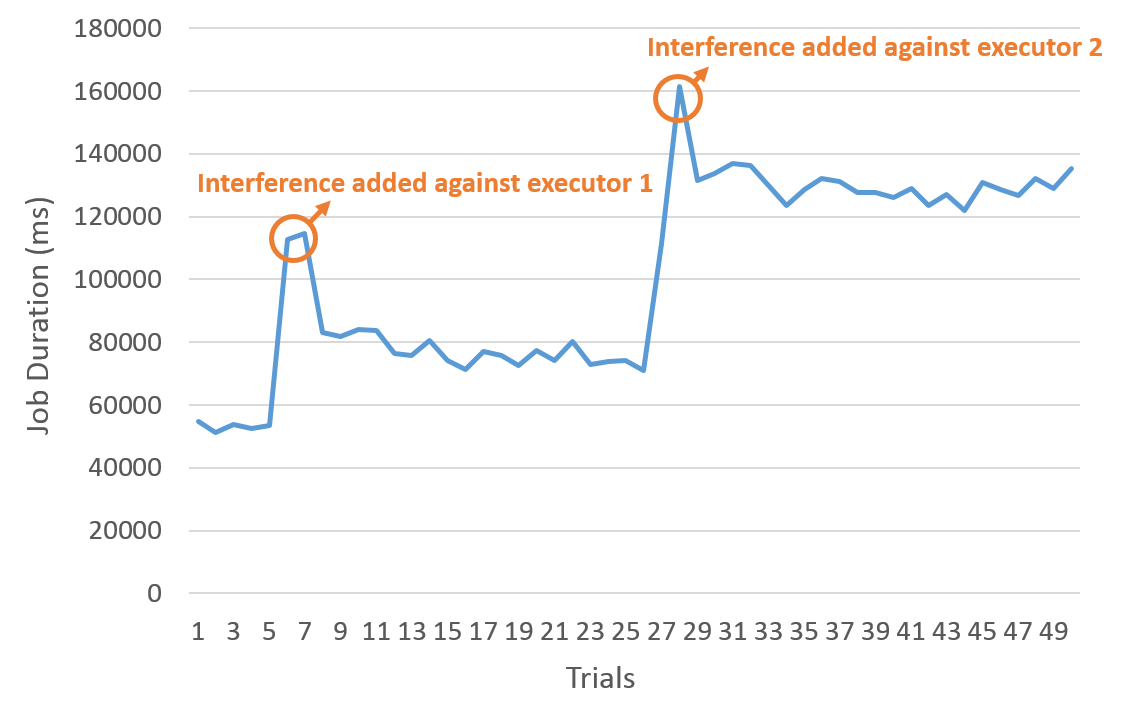}
	\caption{Adaptive workload balancing with introduced
interfering processes at two points in time.}
	\label{workload-rebalance}
\end{figure} 

We performed another experiment involving two hosts being statically
provisioned with one and $0.4$ cores (cf. Sec. \ref{sec:containers}), 
i.e., heterogeneous executors by {\em initial} provisioning.
The results are shown in Fig. \ref{container-rebalance}. Spark learns the optimal way of partitioning the workload after two trials, so the map-stage execution time is reduced to around $60$ seconds, which is in 
agreement  with the results shown in 
Fig. \ref{container04}
where a near optimal data partitioning can be simply derived a priori using resource allocation information provided by Mesos.

\begin{figure}[ht!]
	\centering
	\includegraphics[scale=0.71, trim={0.5cm, 0.5cm, 0, 0}]{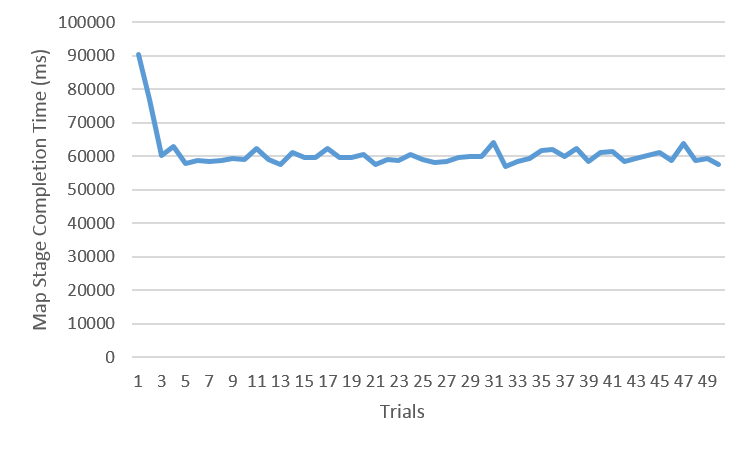}
	\caption{Workload re-balance when executors are different
 by initial provisioning.}
	\label{container-rebalance}
\end{figure} 

This 
online adaptive task sizing can also be applied to the cases in the following Sections \ref{sec:containers} and \ref{sec:burstables}, where the computation capacities of the nodes can be estimated and quantified a priori, and can be further fine-tuned online to achieve better performance.

\section{Heterogeneous MacroTasking (HeMT) with Provisioned Instance Types}\label{sec:HeMT}

In this section, we show how information regarding 
executors (e.g., known resource allocations to executors,
information the service-level agreements)
or some runtime ``state'' of an executor (e.g., token-bucket state) 
can be used to determine initial/baseline heterogeneous task sizes. As
above, task sizes can be further adapted online based on, e.g., execution  time
information received 
from either cluster manager (e.g., Mesos) or monitoring services 
(e.g., AWS CloudWatch).

\subsection{HeMT for Statically Provisioned Containers}\label{sec:containers}

We now compare HeMT with HomT through 
experiments with heterogeneous executors each assigned a different
fraction of a core.
Our implementation supports flexible CPU usage limitation by using containers. 
Baseline Spark does not support partial CPU usage. So we modified the Spark driver to be able to: accept a Mesos offer with partial CPU core; launch an executor using the resources in the offer; and record the actual resources available to this executor so that the driver can use this information to rebalance the workload. 
Additionally, if Spark's executor is spawned on the container with a partial CPU core, we let Spark's executor believe that it has one full core so that it is able to communicate with the driver asking for a task.

To evaluate the performance of HeMT with containerized and statically allocated resources, we did a set of Spark experiments on Mesos. The network bandwidth is large enough ($\sim  600$ Mbps) so that CPU is the only bottleneck. In those experiments, we submitted Spark WordCount jobs with different tasking configurations to a Mesos cluster. WordCount is a simple two-stage Spark job in which most computations are done in the first map stage, so we can determine the effect of load-balancing by observing the execution times of the first stage. Our jobs processed 2GB of data residing on a remote HDFS cluster. For each Spark job, we assign two executors to run associated tasks: one with one full core, the other with partial core \footnote{Spark by default would create one map task for each HDFS block. So to make our experiments start with two tasks, we increase the HDFS block size from 128 MB to 1 GB. We use the same HDFS configuration in Sec. \ref{sec:burstables}. For experiments with usual HDFS configuration, see Sec. \ref{sec:multistage}.}. We used Mesos-supported CFS (complete fair scheduler) bandwidth control \cite{cfs-bandwidth} to limit the CPU usage of the containers.

An example experimental result is shown in 
Fig. \ref{container04}.
As we can see from the red beams, HeMT has good performance since it informs Spark of the CPU allocations and our modified 
Spark balances its workload accordingly.
The U-shaped homogeneous tasking curve is similar to the
results of 
\cite{Ousterhout13}: When tasks sizes are too large (tasks too few in number)
there are sychronization delays. When task sizes are too small
(tasks are too large in number) there is microtasking overhead.
Though \cite{Ousterhout13} gives a rule of thumb,
an optimal homogeneous task size also needs to be learned.

\begin{figure}[ht!]
	\centering
	\includegraphics[scale=0.38, trim={2cm, 0, 0, 0}]{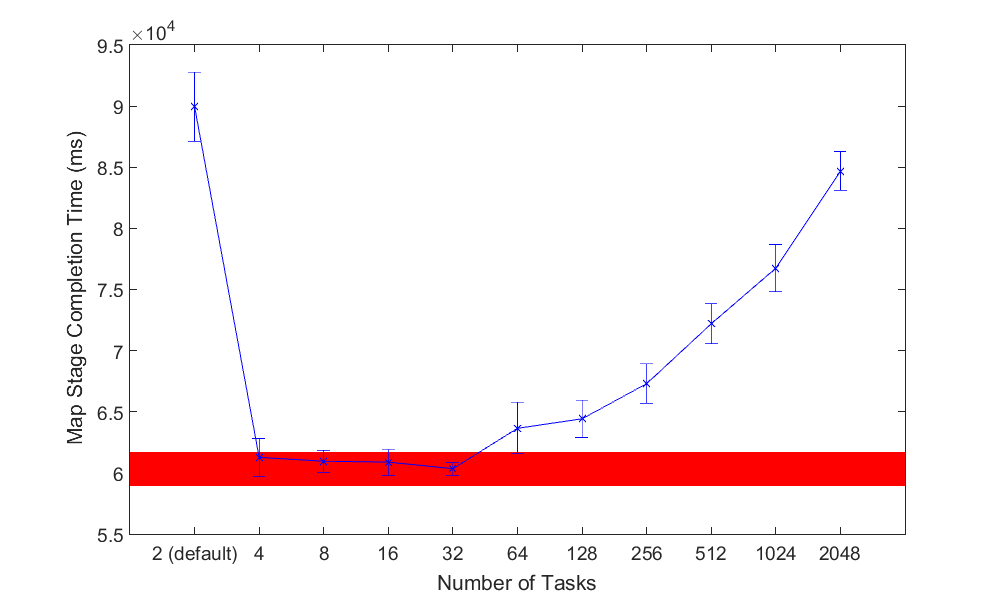}
	\caption{HeMT 
		(red beam for one-$\sigma$ confidence)
		and even partitioning (including HomT and the default 2-way even partitioning). 
		One executor was assigned with $0.4$ CPU and the other a full CPU.}
	\label{container04}
\end{figure}

\subsection{HeMT for Burstable Instances}\label{sec:burstables}

HeMT can also be adapted to support lower cost\footnote{compared to on-demand instances allocated with maximum resources available to the burstable}, general purpose Amazon Web Services (AWS) burstable instance types (T2). Access to CPU 
and network I/O 
\cite{SIGMETRICS17} 
resources  of those instances are governed by a token-bucket mechanism. The CPU credits are earned and spent at a millisecond-level resolution, and the CPU credits accumulated when CPU(s) are idling can be used for future CPU bursts. 
A CPU credit can be used for one CPU running with $100\%$ utilization for one minute,  and it can be divided down to a millisecond timescale
for a short CPU activity burst.
Different types of T2 burstable instances have different CPU credits earning rate, which is in line with their own baseline performance (the performance when an instance has zero credit), and CPU caps (peaks) at which CPU credits stop accumulating. The existing CPU credits indicated through AWS cloudwatch API provide a basis for HeMT workload skewing\footnote{AWS cloudwatch metrics on the free tier update every 5 minutes. If detailed monitoring is enabled
 at extra cost, then AWS cloudwatch update frequency can reach a maximum of  only once per minute \cite{aws-monitoring}. Therefore, AWS cloudwatch may not be helpful for resource-status reporting in short-lived clusters or those with high dynamism.}.

Given the current amounts of CPU credits and the baseline CPU performance, suppose we are able to estimate the computational workload of our job,
the calculation of the amount of work that a node can process within a certain time, $W(t)$, can be easily evaluated: Suppose a t2.small instance initially has $4$ CPU credits. If its vCPU is continually busy, then its CPU credits will be used up in $4 / (1 - 0.2) = 5$ minutes, and, afterwards, its CPU performance will drop to the $20\%$ baseline. The workload it can process in $10$ minutes can be calculated as
\begin{equation*}
W(10) = 1 \times \frac{4}{1 - 0.2} + 0.2 * (10 - \frac{4}{1-0.2}) = 6,
\end{equation*}
reflecting   the size of shaded area in Fig. \ref{mapped-workload}.
\begin{figure}[ht!]
	\centering
	\includegraphics[scale=0.33]{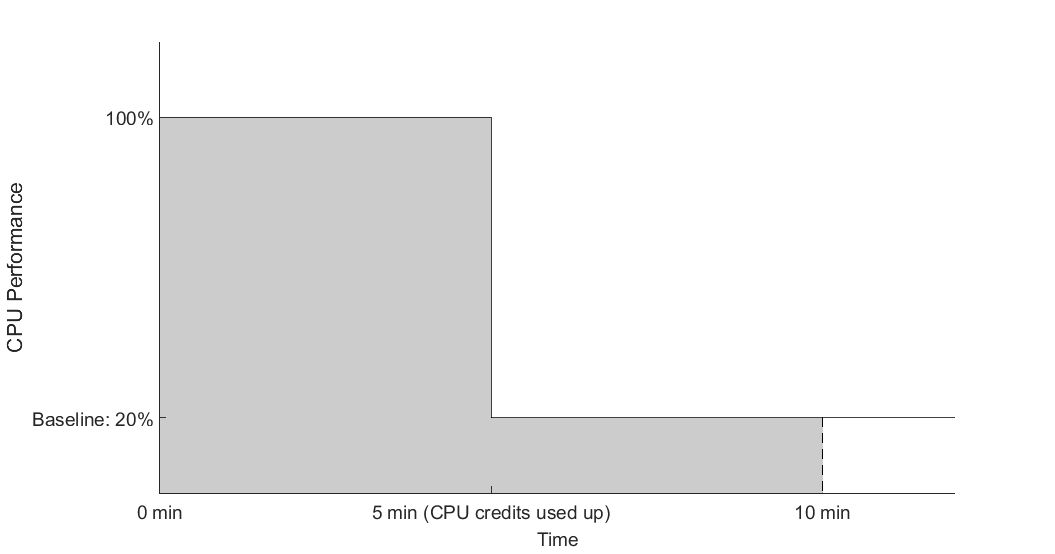}
	\caption{Mapped 10-min workload for a t2.small instance with $4$ initial CPU credits.}
	\label{mapped-workload}
\end{figure}

To divide a workload ($W_0$) to multiple servers with different amount of initial CPU credits so that they can finish at the same time, for each server, we first transform the time-credits plot in Fig. \ref{mapped-workload} into the time-workload plot as shown in Fig. \ref{time-workload}.
\begin{figure}[ht!]
	\centering
	\includegraphics[scale=0.35, trim={2cm, 0, 0, 0}]{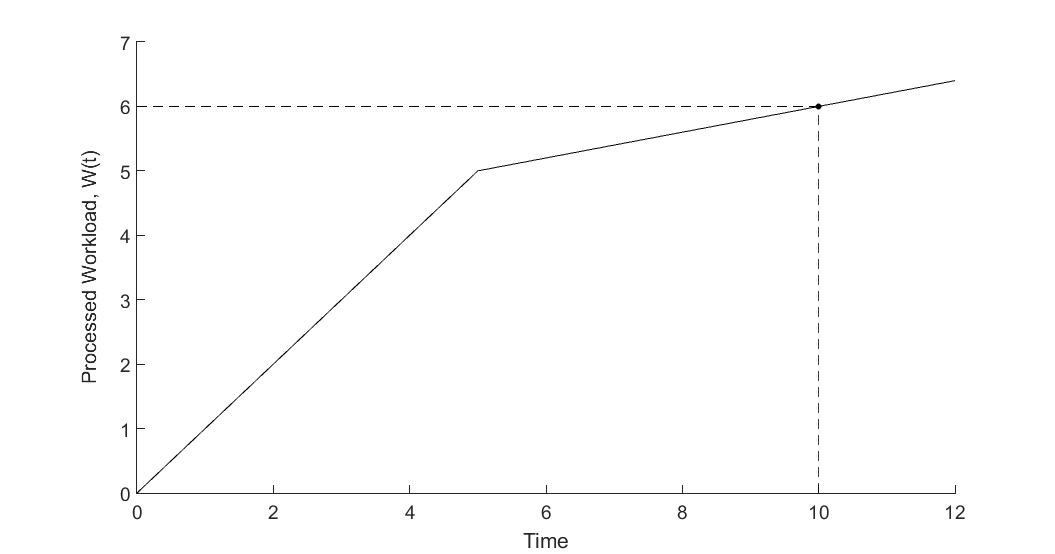}
	\caption{Transformed time v.s. workload plot from Fig. \ref{mapped-workload}.}
	\label{time-workload}
\end{figure}
Superposing such time-workload graphs together into a single piecewise linear function ($\hat{W}(t)$), we can  find $t'$ such that $\hat{W}(t') = W_0$.  
We then divide the workload proportionally according to $\left\lbrace W_i(t') \right\rbrace$, where $i$ is the index of node $i$.

For example, suppose we have three computation nodes with 4, 8, 12 initial CPU credits respectively, and the current data process job requires a CPU running at $100\%$ performance for 20 minutes. We can first superpose the time-workload graphs for these three nodes together as $W_s(t)$, then find $t' = \frac{80}{11}$ such that $W_s(\frac{80}{11}) = 20$, as shown in Fig. \ref{stacked-workload}. Finally we divide the entire workload for the three nodes according to their weights $\left\lbrace W_1(80/11), W_2(80/11), W_3(80/11)\right\rbrace = \{\frac{60}{11}, \frac{80}{11}, \frac{80}{11}\} \propto \{3, 4, 4\}$.
\begin{figure}[ht!]
	\centering
	\includegraphics[scale=0.35, trim={2cm, 0, 0, 0}]{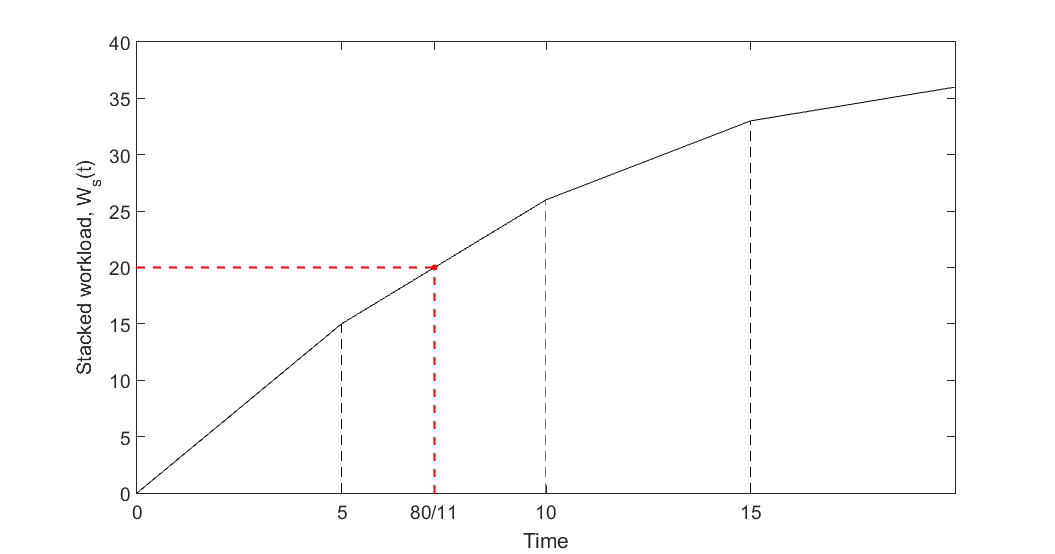}
	\caption{Superposed workload from node 1 to 3.}
	\label{stacked-workload}
\end{figure}

As the CPU credit-monitoring API provided by AWS is not very responsive,
it may not be useful to apply HeMT for short tasks. So, we created the following experimental scenario: First, there are only two types of nodes, one with sufficient CPU credits that will not be depleted throughout the entire life-span of a job, and the other with zero CPU credits\footnote{We can only ensure zero CPU credits on a particular node when starting the job subject to the AWS monitor's update latency.}, where 
\textit{sysbench} \cite{sysbench} was used to deplete CPU credits.
That is, to create heterogeneity, we made depleted one node's CPU credits so 
its CPU works with baseline performance 
(40\% of CPU for AWS t2.medium instance). 
We let our map tasks fetch $2$ GB input data from HDFS. %To make the duration of a single task long enough,
Again to make our experiment start with two tasks, we set the HDFS block size to $1$GB, so each task would  process $1$GB data (one HDFS data block) in Spark's default setting. 

Figs. \ref{100bw}-\ref{40bw}
shows the completion time of the map stage under different configurations. 
This set of experiments was done in a small Spark cluster with two executors, each with one core (on two separate AWS burstable instances). 
The tasks read input from a remote HDFS cluster consisting
of four datanodes, each an AWS t2.small instances.
In these experiment, the network bandwidth is large enough ($\sim  480,600$ Mbps) so that CPU is the only bottleneck.

\begin{figure}[ht!]
\centering
\includegraphics[scale=0.38, trim={2cm, 0, 0, 0}]{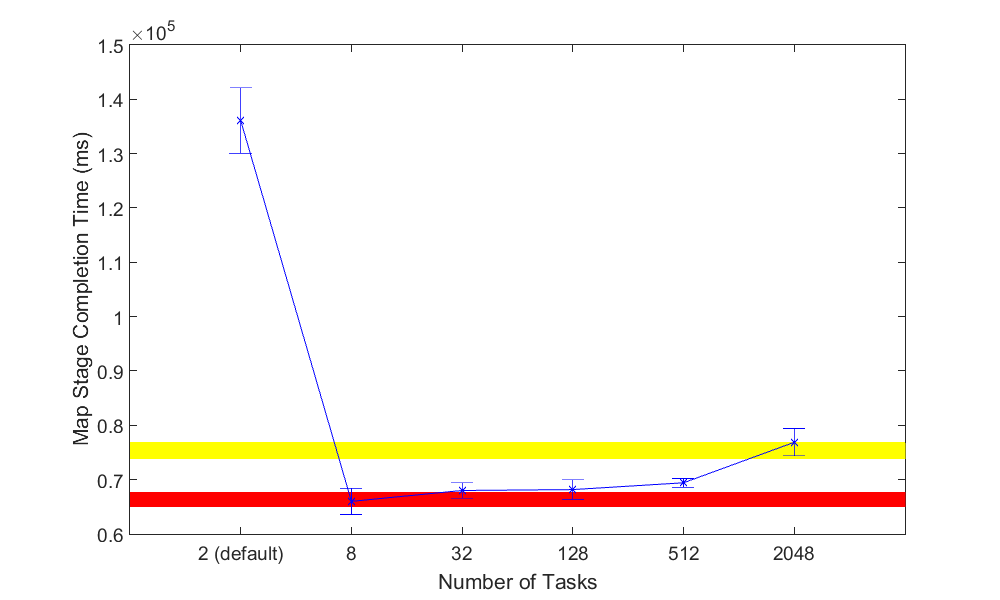}
\caption{Micro-tasking (HomT) vs. macrotasking  (HeMT), 
where HeMT yellow beam is without fudge factor and red beam 
is with fudge factor,
when CPU is the only bottleneck.
}
\label{100bw}
\end{figure}

We also implemented the HeMT approach, i.e., keeping the number of tasks unchanged, but skewing their input data size according to the speed of node on which they are assigned. 
The yellow beams of 
Figs. \ref{100bw} -\ref{40bw}
shows one-$\sigma$ confidence interval of a {\em naive} implementation of workload skewing (HeMT), where we partition the data strictly according to CPU peak and baseline performance (1:0.4 for one core on AWS t2.medium). That is, the task running on the faster node gets $1/1.4$ of the total data, while the other gets the rest. However, we found that the node with zero CPU credit runs even slower than $40\%$ of peak speed. We suspect that this is because the task running with baseline performance is likely facing a higher degree of CPU cache and TLB contention than the other task (e.g., it is possible that the first task is sharing a physical CPU with one or more other workloads while the second task has an entire CPU to itself). Our workload had a significant fraction of memory instructions that are delayed due to such cache/TLB contention. 
By employing short/trial probing tasks,
we found that data partitioning by 1:0.32 further improves load-balancing, i.e.,
this fudge factor is learned from runtime observations (recall
Sec. \ref{sec:adapted-hemt}).
This is shown by the red beam in Fig. \ref{100bw} - the performance of HeMT with this fudge factor improves performance over the best configuration HoMT (8-way) we had tried.

We observed a similar result (with larger variance however) when reduced available bandwidth to $\sim 480$Mbps by using a network traffic shaper, wondershaper \cite{wondershaper}, as shown in Fig. \ref{80bw}. Since CPU remains the bottleneck, we observed thus decreasing network bandwidth does not cause any significant impact on processing speed as expected.

\begin{figure}[ht!]
\centering
\includegraphics[scale=0.38, trim={2cm, 0, 0, 0}]{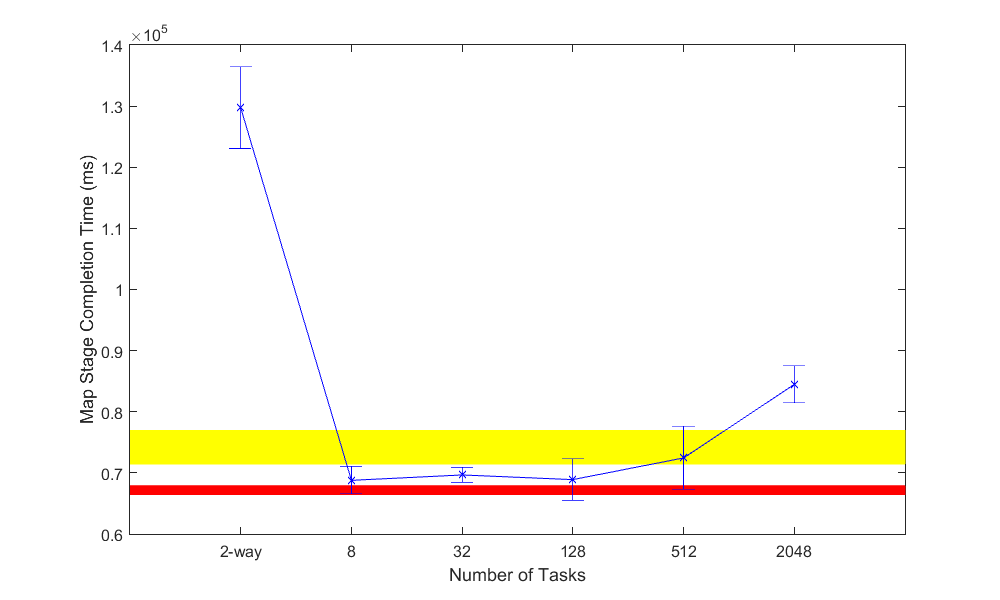}
\caption{Stage completion time when the bandwidth reduced to 
$\sim 480$Mpbs, where CPU is still the only bottleneck.}
\label{80bw}
\end{figure}

We found the behavior of microtasking and macrotasking approaches were different when available uplink bandwidth of a HDFS datanode was reduced to $\sim 250$Mbps, see Fig. \ref{40bw}. 
In this case, for the node with sufficient CPU credits, network I/O becomes the bottleneck, while CPU remains to be the bottleneck for the task running on the node with zero CPU credits. Note that 8-way partitioning is no longer one of the best HomT/microtasking approaches because this relatively coarse-grained partitioning\footnote{where a task running on credit-abundant node would run for about $15$ seconds, while a task running on the other node would run for about $30$ seconds} fails to well-balance the workload in this case. 
It can also be observed that HeMT, including the naive CPU-credit-based partitioning (even it does not make sense in this scenario since the node with sufficient CPU credits is now bottlenecked by the network), started to significantly outperform HomT, because latter is  more likely to incur datanode uplink contention, as explained in Sec. \ref{sec:homt}.

\begin{figure}[ht!]
\centering
\includegraphics[scale=0.38, trim={2cm, 0, 0, 0}]{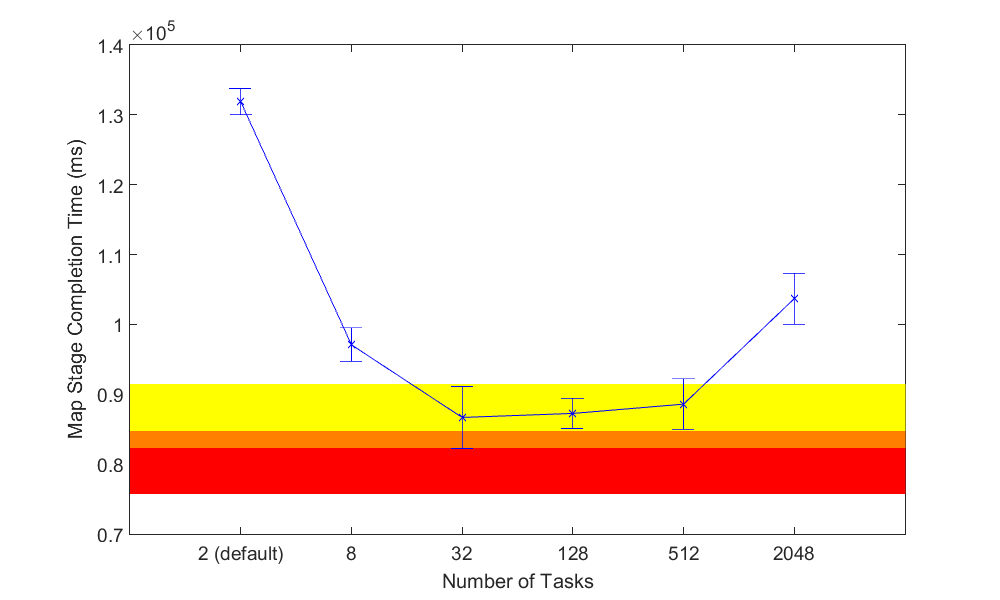}
\caption{Stage completion time when we cut the bandwidth to $\sim 250$Mbps, where the node with sufficient CPU credits becomes bottlenecked by network bandwidth. The orange area represents the overlapping area between the confidence interval of naive CPU credit-based partitioning (yellow) and that of the adjusted partitioning (red).}
\label{40bw}
\end{figure}

\section{HeMT - repartitioning on multiple computation stages}\label{sec:multistage}
 
Our heterogeneous macrotasking can certainly be applied to more realistic workload. A typical MapReduce workload consists of one or more jobs, each job has multiple basic computation stages presented in the previous sections concatenated together through data shuffling. So for the first computation stage, we can simply divide the initial input data according to the computation capacities of the executors. 

A partitioner defines the how a task assigns its intermediate results to different ``buckets'' which will be fetched by different tasks in the following stage respectively. For the following stages, task data are fetched from the intermediate outputs of the tasks in the previous stages. The tasks in the previous stages first shuffle the processed records into different buckets (each corresponding to one fetching task in a future stage) according to a partitioner function, then those buckets are written onto storage media for associated future tasks to fetch. The default hash partitioner shuffles those records into those buckets in a statistically even fashion. So, we need to define a new partitioner that can skew the shuffle buckets for HeMT. We show one implementation of skewing using hash code in Algorithm \ref{skewed_hash}\footnote{Certainly, more sophisticated partitioning algorithm can be made given more information regarding key distribution and processing complexity of each record.}. 

\begin{algorithm}
	\KwData{Record $r$ to be assigned to a bucket; array of executors' computation capacities, $executors$}
	\KwResult{The index of the target bucket}
	$sum$ = $0$\;
	\For{$e$ in $0$ until $executors$.length}{
		$sum$ += $executors$[e]\;
		$executors$[e] = sum\;
	}
	$hash$ = $r$.hashCode mod $executors$.sum\;
	\Return the number of elements in $executors$ greater than or equal to $hash$.
	\caption{Partitioning function of skewed hash partitioner}
	\label{skewed_hash}
\end{algorithm}

The comparison of effective data flows when using the default hash partitioner and our skewed hash partitioner respectively is shown in Fig. \ref{even_vs_skewed}. Relevant idea of balancing workload through partitioner can be found in \cite{lb-mapreduce, adaptive-mapper}.

\begin{figure}[ht!]
	\centering
	\includegraphics[scale=0.19]{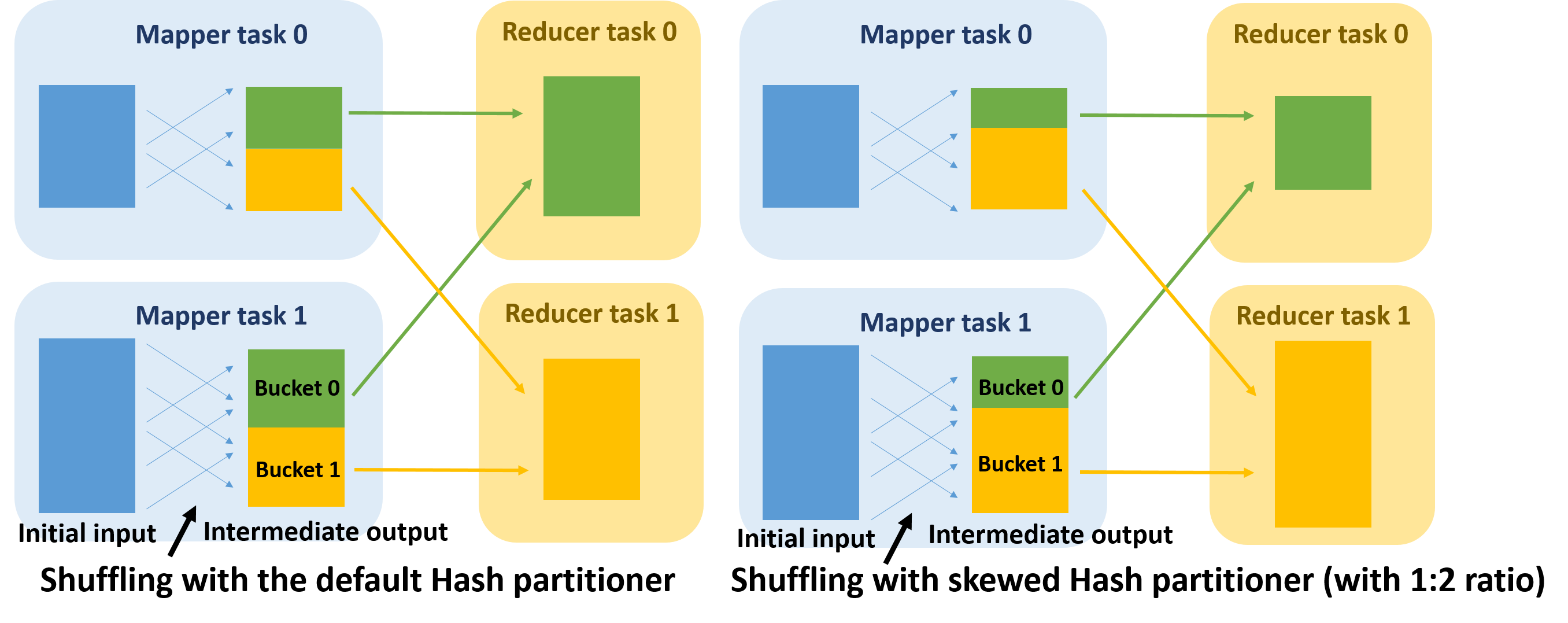}
	\caption{Data flows in even hash shuffling and skewed hash shuffling respectively.}
	\label{even_vs_skewed}
\end{figure}

We present the performance of HeMT using two typical workloads - K-Means and PageRank. Those two have different and representative computation patterns. K-Means consists of repetitive simple two-stage Spark jobs. PageRank, on the other hand, is a single Spark job containing multiple computation stages concatenated together through shuffling.

Again, we run K-Means on the cluster with two executors hosted on two containers, one was allocated with one CPU core, the other was allocated with $0.4$ cores.
To make results more consistent, instead of setting a convergence criterion to stop the iterations, we fix the number of iterations to $30$. The input source is 256 MB data file on HDFS, with block size 128 MB (So there are two blocks). The entire job finish times of HeMT and HoMT are shown in Fig. \ref{kmeans}, which is consistent with the single-stage results shown in the previous sections.

\begin{figure}[ht!]
	\centering
	\includegraphics[scale=0.38, trim = {2cm, 0, 0, 0}, clip]{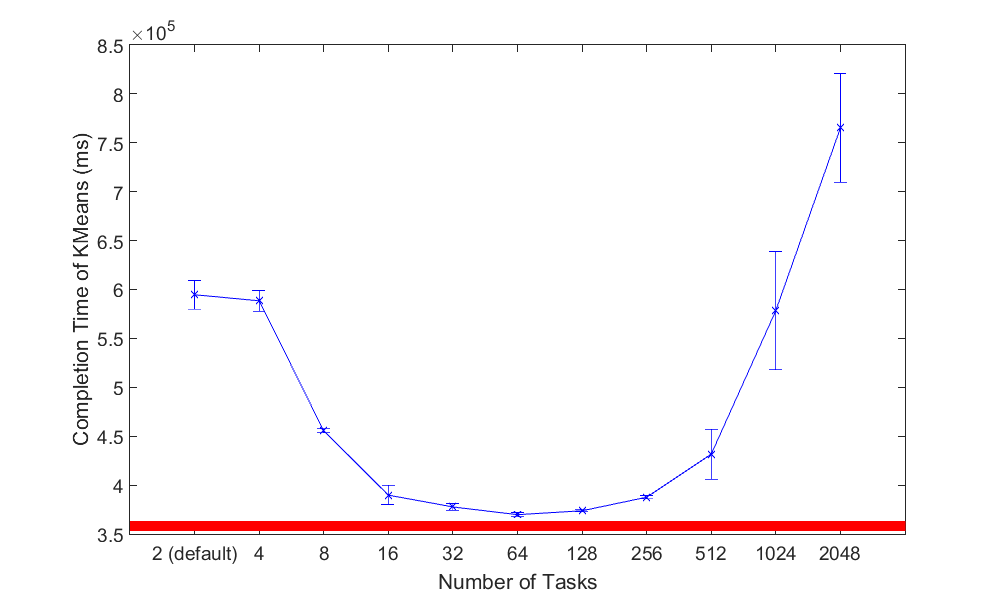}
	\caption{K-Means finish time.}
	\label{kmeans}
\end{figure}

On the same cluster, we run PageRank with $256$ MB input data for $100$ iterations. The results are shown in Fig. \ref{pagerank}. Note that the PageRank, compared with K-Means, is more sensitive to microtasking, because each iteration of PageRank is relatively short (around 10s in the default 2-way partitioning), therefore each task is shorter as well. For example, if we use 64-way partitioning, then each task generally lasts for only $0.1$ - $0.2$ seconds. Therefore, the relative task scheduling overhead would be larger in PageRank workload.
\begin{figure}[ht!]
	\centering
	\includegraphics[scale=0.38, trim = {2cm, 0, 0, 0}, clip]{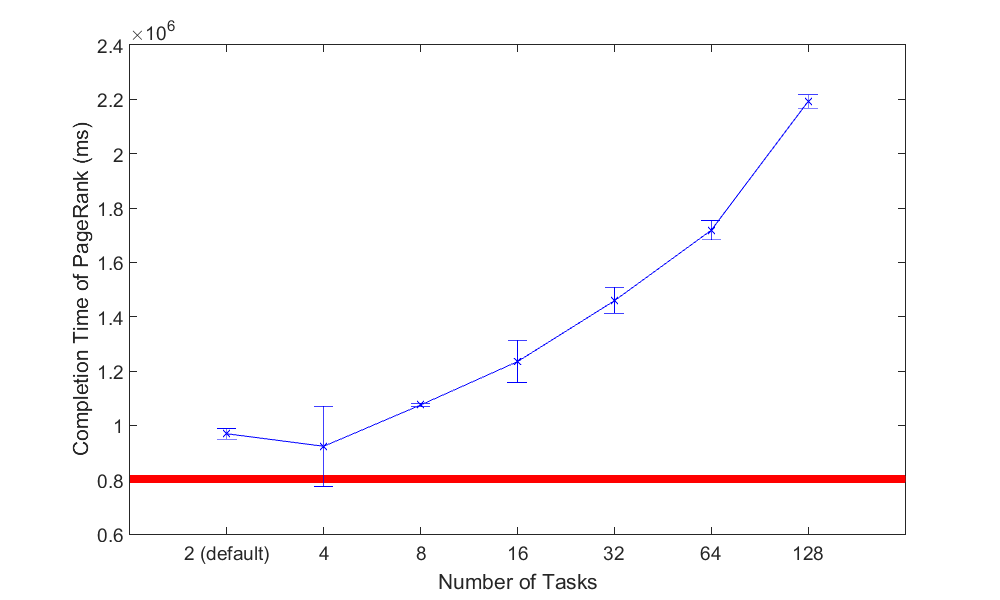}
	\caption{PageRank finish time.}
	\label{pagerank}
\end{figure}

\section{Related Work} \label{sec:related}

We discuss related work, mainly in the context of recent developments in cloud computing, in this section.
Pointers to older work may be found in the references within the papers we discuss here. 

\noindent {\bf Parallel scheduling advances including straggler mitigation: }  Computational skew
(heterogeneous tasking) has been investigated in several studies in order to mitigate the straggler problem.
\cite{adaptive-mapper} focuses on the map stage in a mapreduce job. It breaks the one-to-one mapping between
a mapper task and data segments in Hadoop. By allowing mapper tasks to continue fetching other data segments,
more data ends up being fetched by the faster mappers. In this way, the workload is automatically balanced
especially when it is  finely partitioned. Unlike the classic mapreduce implementations in Hadoop and Spark
where mappers are running independently, this approach requires mappers to be synchronized on read during
runtime, which could affect its scalability.

\cite{skewreduce, skewtune} also target computational skew. They consider a homogeneous cluster where computational skew is the cause of the straggler problem. 
Even data segments of the same size may require different processing time by the same executor.  With good estimates of processing complexity of partitions, the straggling partitions are split in either static \cite{skewreduce} or dynamic \cite{skewtune} fashion.

\cite{loc-based-data-partitioning} showed that enlarging task size by merging multiple co-located block into one task can save task scheduling and initialization times and thereby accelerate the execution of Hadoop.
Some studies and practices such as 
\cite{spark-config,Ousterhout13, mantri, grass} 
take an opportunistic approach by letting the driver employ time-out at program 
barriers to detect
straggler tasks and relaunch them on new executors (speculative execution).

Other studies rely on detailed monitoring and performance predictions to mitigate the straggler problem. For example, \cite{wrangler} predicts a straggler node by applying a support vector machine on features involving resource utilization, thread states and memory statistics. \cite{wrangler} then conservatively 
prevents tasks from being assigned on those nodes that
 are predicted to be stragglers.

 Much prior work has been conducted on  the
representation and prediction of resource capacity demand 
(workload characterization) and supply (executor characterization).
Also, feasible scheduling decisions,
by both application frameworks and the tenant's middleware
cluster managers, have been proposed for more efficient use of
the tenant's available resources, cf. Sec. \ref{sec:related}.
Online learning approaches have also been proposed to this end,
e.g., \cite{Alizadeh18}. Orchestration of budget conscious tenants
using low-cost instances has also been an active area of study, particularly
through the use of preemptible spot and/or
burstable/bursting instances.
Finally, there are several different open-source application frameworks
and cluster managers which can be modified to experiment
with  more advanced learning and scheduling mechanisms.

\noindent {\bf Representation and prediction of resource capacity demand and supply:} Workload prediction for individual applications is an area of extensive research in cloud-like environments, e.g., 
\cite{770691,Hellerstein:2001:SAP:370802.370814,Jalaparti:2012:BTG:2391229.2391239},
as is the related problem of translating workload-specific predictions (e.g., request arrivals) into resource allocations (e.g., servers) to meet desired performance goals, e.g., 
\cite{1687544,sigmetrics05,Doyle:2003:MRP:1251460.1251465,Chandra:2003:DRA:781027.781067,Levy2003,N.Bennani:2005:RAA:1078027.1078472}.

For example, suppose a sequence of roughly equal sized
data batches need to be processed. For each data batch,
suppose there are a certain number of tasks, each of
one of say $K$ different types (and of
a specific size) that have to be executed
in some order (including some in parallel).
At any given point in time, suppose the
``state" of the application is given by
(i) the type of task running in each of its executors and its execution time so far, and (iii) any queued tasks.
A reinforcement learning approach, e.g., \cite{Alizadeh18}, could train
a neural network to 
 map the system state to scheduling decisions (which task to assign to
the next available executor)
in order to minimize overall execution time for the data batch
(to minimize synchronization delays at 
current program barriers/joins in particular).
For several initial data batches, random scheduling decisions could
be  made and their consequences in terms of 
execution time observed and used for training.
Here, the resource requirements of each task to meet certain performance
criteria is not considered.
%Yuquan:
%The scheduling improvement is achieved through training on the previous results.The authors didn't provide any concrete reason (in terms of system-wise factors such as data locality, CPU contention, etc) of this improvement.  By the way, without giving their reward/cost function, it is hard to know how they obtain the learning function in Eq. (1).

\noindent
    {\bf Workload admission control and resource allocation:} This has also been extensively
    studied, 
typically assuming some knowledge of (IT) resources demand 
(and their correlations) and supply.
Such problems are related to
multidimensional knapsack, bin-packing and load-balancing problems, e.g.,
 \cite{CK04,PRP10,Varnamkhasti12,CKPT16},
which are NP-hard. They have been extensively studied, including relaxations to
simplified problems that yield approximately optimal
solutions, e.g., by Integer Linear Programs solved by
iterated/online means, and including stochastic versions, e.g.,
\cite{Roytman13,Cohen17}.
 A relevant sub-field, starting with dominant-resource fairness (DRF)~
\cite{DRF},  has been concerned with extending classical single-resource, single server {\bf fair} scheduling to multiple resources and multiple servers as appropriate for highly utilized cloud environments e.g., 
\cite{HUG,Friedman14,BLi16b}.
Most of this scheduling work attempts to incrementally maximize performance while optimally utilizing shared resource between applications/jobs 
\cite{Lu:2015:VLF:2806777.2806943,Curino:2014:RSY:2670979.2670981,Tembey:2014:MAP:2670979.2670993,wrangler,Shafiee:2017:SCD:3078505.3078548, Delimitrou-quasar,Delimitrou:2015:TRS:2806777.2806779}.

\noindent {\bf Orchestration of budget-conscious tenants in the public cloud:} This area has been receiving significant attention, particularly in the form of using Amazon's spot instances (spanning concerns of cost-efficacy, resource availability and optimal bidding),
e.g., 
\cite{sharma2015spotcheck,subramanya2015spoton};
using resource reservations, e.g., 
\cite{yao2014online,Delimitrou16};
exploiting price/performance/availability trade-offs across geo-distributed sites (including multiple cloud markets), e.g., 
\cite{wu2013spanstore,nawab2015minimizing,lucas2013scheduling};
design and analysis of cloud aggregators/brokers and ``derivative'' clouds, e.g.,
\cite{aazam2014broker,liu2014dynamic}. 
Attention has also been given to exploiting the relatively new burstable instances, e.g., \cite{sigmetrics17,EuroSys17}. 

\noindent{\bf Open-source cluster management frameworks: } 
Whereas numerous cluster orchestration frameworks exist as open-source software~
\cite{mesos,kubernetes,Vavilapalli:2013:AHY:2523616.2523633,omar,marathon,swarm}, 
%Sapphire? https://github.com/UWSysLab/Sapphire
%YS: Sapphire, likes Spark, can be used to distribute tasks to multiple machines. However, unlike Spark and Mesos, there is no centralized scheduler for the cluster. So the programmer himself has to act as Mesos resource scheduler and Spark task scheduler to manually pick which machine (Mesos' job) to run what (Spark's job) (by selecting deployment manager (DM) and defining Sapphire objects (SO)).
they tend to implement resource management policies suitable for private settings. Public cloud orchestration frameworks are naturally proprietary in nature. In the following, we leverage the existing code-base of Apache Mesos~\cite{mesos} that
employs a default scheduling mechanism   DRF
\cite{DRF} 
%(which was devised for a single server) 
adapted to a cluster of servers.

\section{Future Work}\label{sec:fw}

In future work, we will consider
application  frameworks and middleware embodying
more advanced, integrated  online learning frameworks that leverage
information from
offline workload profiling and service-level agreements
 to more precisely 
(online) characterize
workloads' resource needs (demand) and executors' capacity (supply).
Actions by different application frameworks
based on such learning include HeMT at a fast timescale
and determination of preferred types of executors based on
cost/performance tradeoffs.
For a budget-conscious tenant,
we also plan to integrate such actions 
by application frameworks with 
scheduling by the cluster manager (middleware), i.e.,
\cite{DRF} and server-specific alternatives \cite{Jalal17a,Shan18},
the latter improving the efficiency of resource use.
That is, the cluster manager's scheduler would be based on 
online estimates of the resource needs of tasks of its application
frameworks in order to obtain adequate performance
(``fine grain'' mode scheduling).

\section{Concluding Remarks}
\label{sec:conclus}

We investigated the pros and cons of two opposite views on load balancing - homogeneous
microtasking and heterogeneous macrotasking - in large-scale parallel processing workloads
that routinely run on modern public cloud platforms. Using tiny, equal-sized tasks
(homogeneous microtasking, HomT) has long been regarded as an effective way of
load balancing in parallel computing systems. When combined with nodes pulling in work
upon becoming idle,  more powerful nodes finish their work sooner and, therefore, pull in
additional work faster. As a result, HomT is deemed especially desirable in settings with
heterogeneous (and possibly possessing dynamically changing) processing capacities.
However, HomT poses additional scheduling and I/O overheads 
that may make it more costly in some scenarios. In this paper,
we analyzed these advantages and disadvantages of HomT using a combination of analytical
modeling and data-driven experiments on an Apache Spark based prototype (Spark's built-in
scheduler, when parameterized appropriately, already implements HomT.) 
We then proposed an alternative load balancing scheme -  Heterogeneous Macrotasking (HeMT),
wherein workload (input data) was {\em intentionally and carefully} partitioned into tasks of possibly
different sizes. We explored different heuristics for such heterogeneous partitioning
based on the amount of information available about the nodes' processing capacities. 
Our goal was to study when HeMT could overcome the performance disadvantages of HomT.
We implemented a prototype of HeMT within the Apache Spark application framework with
complementary enhancements to the Apache Mesos cluster manager.  Our experimental results
showed that HeMT outperformed HomT when accurate workload-specific estimates of nodes'
processing capacities could be learned. In our experiments,  Spark with HeMT was able to
improve average job completion times 
by about 10\% compared to the default system. 

\section*{Acknowledgements:}
This research was supported in part by NSF CNS 1717571 grant and a Cisco Systems URP gift. 

\bibliographystyle{plain}
\bibliography{./bibfiles/scheduling,./bibfiles/cloud,./bibfiles/opt,./bibfiles/mars,./bibfiles/refs-cloud,./bibfiles/kesidis-prior,./bibfiles/IPM,./bibfiles/ref_cheng,./bibfiles/cloud2,./bibfiles/knapsack,./bibfiles/games,./bibfiles/ratio,./bibfiles/pricing-cloud,./bibfiles/proposal-csr,./bibfiles/eurosys19}

\def\cprime{$'$}
\begin{thebibliography}{10}

\bibitem{sysbench}
{A. Kopytov}.
\newblock Sysbench.
\newblock \url{https://github.com/akopytov/sysbench}, 2016.

\bibitem{aazam2014broker}
M.~Aazam and E.-N. Huh.
\newblock {Broker as a Service (BaaS) pricing and resource estimation model}.
\newblock In {\em Proc. IEEE CloudCom}, 2014.

\bibitem{raas}
Orna Agmon Ben-Yehuda, Muli Ben-Yehuda, Assaf Schuster, and Dan Tsafrir.
\newblock The resource-as-a-service (raas) cloud.
\newblock In {\em Proceedings of the 4th USENIX Conference on Hot Topics in
  Cloud Ccomputing}, HotCloud'12, pages 12--12, Berkeley, CA, USA, 2012. USENIX
  Association.

\bibitem{aws-monitoring}
{Amazon Web Services}.
\newblock Amazon cloudwatch pricing.
\newblock https://aws.amazon.com/cloudwatch/pricing/.
\newblock [Online; accessed 21-Jun-2018].

\bibitem{grass}
Ganesh Ananthanarayanan, Michael Chien-Chun Hung, Xiaoqi Ren, Ion Stoica, Adam
  Wierman, and Minlan Yu.
\newblock Grass: Trimming stragglers in approximation analytics.
\newblock In {\em Proceedings of the 11th USENIX Conference on Networked
  Systems Design and Implementation}, NSDI'14, pages 289--302, Berkeley, CA,
  USA, 2014. USENIX Association.

\bibitem{mantri}
Ganesh Ananthanarayanan, Srikanth Kandula, Albert Greenberg, Ion Stoica, Yi~Lu,
  Bikas Saha, and Edward Harris.
\newblock Reining in the outliers in map-reduce clusters using mantri.
\newblock In {\em Proceedings of the 9th USENIX Conference on Operating Systems
  Design and Implementation}, OSDI'10, pages 265--278, Berkeley, CA, USA, 2010.
  USENIX Association.

\bibitem{hdfs-design}
{Apache Hadoop}.
\newblock Hdfs design.
\newblock
  https://hadoop.apache.org/docs/current/hadoop-project-dist/hadoop-hdfs/HdfsDesign.html.
\newblock [Online; accessed 20-Jul-2018].

\bibitem{hdfs-rack-awareness}
{Apache Hadoop}.
\newblock Hdfs rack awareness.
\newblock
  https://hadoop.apache.org/docs/current/hadoop-project-dist/hadoop-common/RackAwareness.html.
\newblock [Online; accessed 28-Sep-2018].

\bibitem{N.Bennani:2005:RAA:1078027.1078472}
M.N. Bennani and D.A. Menasce.
\newblock Resource allocation for autonomic data centers using analytic
  performance models.
\newblock In {\em Proc. IEEE ICAC}, 2005.

\bibitem{Chandra:2003:DRA:781027.781067}
A.~Chandra, W.~Gong, and P.~Shenoy.
\newblock Dynamic resource allocation for shared data centers using online
  measurements.
\newblock In {\em Proc.ACM SIGMETRICS}, 2003.

\bibitem{CK04}
C.~Chekuri and S.~Khanna.
\newblock On multi-dimensional packing problems.
\newblock {\em SIAM Journal of Computing}, 33(4):837--851, 2004.

\bibitem{HUG}
M.~Chowdhury, Z.~Liu, A.~Ghodsi, and I.~Stoica.
\newblock {HUG: Multi-resource fairness for correlated and elastic demands}.
\newblock In {\em Proc. USENIX NSDI}, March 2016.

\bibitem{CKPT16}
H.I. Christensen, A.~Khan, S.~Pokutta, and P.~Tetali.
\newblock {Multidimensional Bin Packing and Other Related Problems: A Survey}.
\newblock https://people.math.gatech.edu/$\sim$tetali/PUBLIS/CKPT.pdf, 2016.

\bibitem{Cohen17}
M.C. Cohen, V.Mirrokni, P.~Keller, and M.~Zadimoghaddam.
\newblock {Overcommitment in Cloud Services Bin packing with Chance
  Constraints}.
\newblock In {\em Proc. ACM SIGMETRICS}, Urbana-Campaign, IL, June 2017.

\bibitem{Curino:2014:RSY:2670979.2670981}
C.~Curino, D.~E. Difallah, C.~Douglas, S.~Krishnan, R.~Ramakrishnan, and
  S.~Rao.
\newblock Reservation-based scheduling: If you're late don't blame us!
\newblock In {\em Proc. SOCC}, 2014.

\bibitem{Delimitrou-quasar}
C.~Delimitrou and C.~Kozyrakis.
\newblock {Quasar: Resource-efficient and QoS-aware Cluster Management}.
\newblock In {\em Proc. ASPLOS}, 2014.

\bibitem{Delimitrou16}
C.~Delimitrou and C.~Kozyrakis.
\newblock {HCloud: Resource-Efficient Provisioning in Shared Cloud Systems}.
\newblock In {\em Proc. ASPLOS}, Atlanta, 2016.

\bibitem{Delimitrou:2015:TRS:2806777.2806779}
C.~Delimitrou, D.~Sanchez, and C.~Kozyrakis.
\newblock {Tarcil: Reconciling Scheduling Speed and Quality in Large Shared
  Clusters}.
\newblock In {\em Proc. ACM SoCC}, 2015.

\bibitem{Doyle:2003:MRP:1251460.1251465}
R.P. Doyle, J.S. Chase, O.M. Asad, W.~Jin, and A.M. Vahdat.
\newblock Model-based resource provisioning in a web service utility.
\newblock In {\em Proc. USENIX USITS}, 2003.

\bibitem{duda-hart-stork-01}
R.O. Duda, P.E. Hart, and D.G. Stork.
\newblock {\em Pattern Classification, 2nd Ed.}
\newblock Wiley, 2001.

\bibitem{lb-mapreduce}
Y.~Fan, W.~Wu, D.~Qian, Y.~Xu, and W.~Wei.
\newblock Load balancing in heterogeneous mapreduce environments.
\newblock In {\em Proc. IEEE HPCC \& EUC}, 2013.

\bibitem{Friedman14}
E.~Friedman, A.~Ghodsi, and C.-A. Psomas.
\newblock Strategyproof allocation of discrete jobs on multiple machines.
\newblock In {\em Proc. ACM Conf. on Economics and Computation}, 2014.

\bibitem{DRF}
A.~Ghodsi, M.~Zaharia, B.~Hindman, A.~Konwinski, S.~Shenker, and I.~Stoica.
\newblock Dominant resource fairness: Fair allocation of multiple resource
  types.
\newblock In {\em Proc. USENIX NSDI}, 2011.

\bibitem{workload-enterprise-dc}
D.~Gmach, J.~Rolia, L.~Cherkasova, and A.~Kemper.
\newblock Workload analysis and demand prediction of enterprise data center
  applications.
\newblock In {\em 2007 IEEE 10th International Symposium on Workload
  Characterization}, pages 171--180, Sept 2007.

\bibitem{770691}
J.~L. Hellerstein, Fan Zhang, and P.~Shahabuddin.
\newblock An approach to predictive detection for service management.
\newblock In {\em Integrated Network Management VI. Distributed Management for
  the Networked Millennium. Proceedings of the Sixth IFIP/IEEE International
  Symposium on Integrated Network Management. (Cat. No.99EX302)}, 1999.

\bibitem{Hellerstein:2001:SAP:370802.370814}
Joseph~L. Hellerstein, Fan Zhang, and Perwez Shahabuddin.
\newblock A statistical approach to predictive detection.
\newblock {\em Comput. Netw.}, 35(1):77--95, 2001.

\bibitem{mesos}
B.~Hindman, A.~Konwinski, M.~Zaharia, A.~Ghodsi, A.D. Joseph, R.~Katz,
  S.~Shenker, and I.~Stoica.
\newblock {Mesos: A Platform for Fine-grained Resource Sharing in the Data
  Center}.
\newblock In {\em Proc. USENIX NSDI}, 2011.

\bibitem{Jalaparti:2012:BTG:2391229.2391239}
V.~Jalaparti, H.~Ballani, P.~Costa, T.~Karagiannis, and A.~Rowstron.
\newblock Bridging the tenant-provider gap in cloud services.
\newblock In {\em Proc. SoCC}, 2012.

\bibitem{workload-windows-servers}
S.~Kavalanekar, B.~Worthington, Qi~Zhang, and V.~Sharda.
\newblock Characterization of storage workload traces from production windows
  servers.
\newblock In {\em 2008 IEEE International Symposium on Workload
  Characterization}, pages 119--128, Sept 2008.

\bibitem{Jalal17a}
J.~Khamse-Ashari, I.~Lambadaris, G.~Kesidis, B.~Urgaonkar, and Y.Q. Zhao.
\newblock {Per-Server Dominant-Share Fairness (PS-DSF): A Multi-Resource Fair
  Allocation Mechanism for Heterogeneous Servers}.
\newblock In {\em Proc. IEEE ICC, Paris}, May 2017.

\bibitem{kubernetes}
Kubernetes.
\newblock {Production-grade Container Orchestration}.
\newblock \url{http://kubernetes.io/}, 2016.

\bibitem{skewreduce}
YongChul Kwon, Magdalena Balazinska, Bill Howe, and Jerome Rolia.
\newblock Skew-resistant parallel processing of feature-extracting scientific
  user-defined functions.
\newblock In {\em Proceedings of the 1st ACM Symposium on Cloud Computing},
  SoCC '10, pages 75--86, New York, NY, USA, 2010. ACM.

\bibitem{skewtune}
YongChul Kwon, Magdalena Balazinska, Bill Howe, and Jerome Rolia.
\newblock Skewtune: Mitigating skew in mapreduce applications.
\newblock In {\em Proceedings of the 2012 ACM SIGMOD International Conference
  on Management of Data}, SIGMOD '12, pages 25--36, New York, NY, USA, 2012.
  ACM.

\bibitem{Levy2003}
R.~Levy, J.~Nagarajarao, G.~Pacifici, M.~Spreitzer, A.~Tantawi, and A.~Youssef.
\newblock Performance management for cluster based web services.
\newblock In Germ{\'a}n Goldszmidt and J{\"u}rgen Sch{\"o}nw{\"a}lder, editors,
  {\em Integrated Network Management VIII: Managing It All}, pages 247--261.
  Springer US, 2003.

\bibitem{cfs-bandwidth}
{Linux Kernel}.
\newblock Cfs bandwidth control.
\newblock https://www.kernel.org/doc/Documentation/scheduler/sched-bwc.txt.
\newblock [Online; accessed 15-Aug-2018].

\bibitem{liu2014dynamic}
K.~Liu, J.~Peng, W.~Liu, P.~Yao, and Z.~Huang.
\newblock Dynamic resource reservation via broker federation in cloud service:
  A fine-grained heuristic-based approach.
\newblock In {\em Proc. IEEE GLOBECOM}, 2014.

\bibitem{Lu:2015:VLF:2806777.2806943}
Hui Lu, Brendan Saltaformaggio, Ramana Kompella, and Dongyan Xu.
\newblock {vFair: Latency-aware Fair Storage Scheduling via per-IO Cost-based
  Differentiation}.
\newblock In {\em Proc. ACM SoCC}, 2015.

\bibitem{lucas2013scheduling}
Jose~Luis Lucas-Simarro, Rafael Moreno-Vozmediano, Ruben~S Montero, and
  Ignacio~M Llorente.
\newblock Scheduling strategies for optimal service deployment across multiple
  clouds.
\newblock {\em Future Generation Computer Systems}, 29(6):1431--1441, 2013.

\bibitem{Alizadeh18}
H.~Mao, M.~Schwarzkopf, S.B. Venkatakrishnan, and M.~Alizadeh.
\newblock Learning graph-based cluster scheduling algorithms.
\newblock In {\em Proc. SysML}, Stanford, CA, USA, Feb. 2018.

\bibitem{marathon}
{Marathon - A Container Orchestration Framework for Mesos and DC/OS}.
\newblock \url{https://mesosphere.github.io/marathon/}, last accessed, Sept.
  2017.

\bibitem{mesos-code}
Mesos multi-scheduler.
\newblock https://github.com/yuquanshan/mesos/tree/multi-scheduler.

\bibitem{Roytman13}
A.~Meyerson, A.~Roytman, and B.~Tagiku.
\newblock Online multidimensional load balancing.
\newblock In {\em Proc. Approximation, Randomization, and Combinatorial
  Optimization Algorithms and Techniques - Springer LNCS vol 8096}, 2013.

\bibitem{nawab2015minimizing}
F.~Nawab, V.~Arora, D.~Agrawal, and A.~El~Abbadi.
\newblock Minimizing commit latency of transactions in geo-replicated data
  stores.
\newblock In {\em Proc. ACM SIGMOD}.

\bibitem{flat-dc}
E.B. Nightingale, J.~Elson, J.~Fan, O.~Hofmann, J.~Howell, and Y.~Suzue.
\newblock Flat datacenter storage.
\newblock In {\em Proc. USENIX OSDI}, Hollywood, CA, 2012.

\bibitem{Ousterhout13}
K.~Ousterhout, A.~Panda, J.~Rosen, S.~Venkataraman, R.~Xin, S.~Ratnasamy,
  S.~Shenker, and I.~Stoica.
\newblock The case for tiny tasks in compute clusters.
\newblock In {\em Proc. USENIX HotOS}, 2013.

\bibitem{PRP10}
J.~Puchinger, G.R. Raidl, and U.~Pferschy.
\newblock The multidimensional knapsack problem: Structure and algorithms.
\newblock {\em INFORMS Journal on Computing}, 22(2):250--265, Spring 2010.

\bibitem{omar}
O.~Sefraoui, M.~Aissaoui, and M.~Eleuldj.
\newblock Openstack: Toward an open-source solution for cloud computing.
\newblock {\em International Journal of Computer Applications}, 55(3), 2012.

\bibitem{Shafiee:2017:SCD:3078505.3078548}
M.~Shafiee and J.~Ghaderi.
\newblock {Scheduling Coflows in Datacenter Networks: Improved Bound for Total
  Weighted Completion Time}.
\newblock In {\em Proc. ACM SIGMETRICS}, 2017.

\bibitem{Shan18}
Y.~Shan, A.~Jain, G.~Kesidis, B.~Urgaonkar, J.~Khamse-Ashari, and
  I.~Lambadaris.
\newblock Scheduling distributed resources in heterogeneous private clouds.
\newblock In {\em Proc. IEEE MASCOTS}, Milwaukee, Sept. 2018.

\bibitem{sharma2015spotcheck}
P.~Sharma, S.~Lee, T.~Guo, D.~Irwin, and P.~Shenoy.
\newblock {Spotcheck: Designing a derivative IAAS cloud on the spot market}.
\newblock In {\em Proc. EuroSys}, 2015.

\bibitem{spark-config}
{Apache Spark - Spark Configuration}.
\newblock https://spark.apache.org/docs/latest/configuration.html.

\bibitem{spark-code}
Spark with resource demand vectors.
\newblock https://github.com/yuquanshan/spark/tree/d-vector.

\bibitem{subramanya2015spoton}
S.~Subramanya, T.~Guo, P.~Sharma, D.~Irwin, and P.~Shenoy.
\newblock {SpotOn: A batch computing service for the spot market}.
\newblock In {\em Proc. ACM Symp. on Cloud Computing}, pages 329--341, 2015.

\bibitem{swarm}
{Docker Swarm}.
\newblock \url{https://docs.docker.com/engine/swarm/}, last accessed, Sept.
  2017.

\bibitem{Tembey:2014:MAP:2670979.2670993}
P.~Tembey, A.~Gavrilovska, and K.~Schwan.
\newblock Merlin: Application- and platform-aware resource allocation in
  consolidated server systems.
\newblock In {\em Proc. ACM SOCC}, 2014.

\bibitem{fb-data-warehouse}
Ashish Thusoo, Zheng Shao, Suresh Anthony, Dhruba Borthakur, Namit Jain,
  Joydeep Sen~Sarma, Raghotham Murthy, and Hao Liu.
\newblock Data warehousing and analytics infrastructure at facebook.
\newblock In {\em Proceedings of the 2010 ACM SIGMOD International Conference
  on Management of Data}, SIGMOD '10, pages 1013--1020, New York, NY, USA,
  2010. ACM.

\bibitem{against-tiny}
Ehsan Totoni, Subramanya~R. Dulloor, and Amitabha Roy.
\newblock A case against tiny tasks in iterative analytics.
\newblock In {\em Proceedings of the 16th Workshop on Hot Topics in Operating
  Systems}, HotOS '17, pages 144--149, New York, NY, USA, 2017. ACM.

\bibitem{sigmetrics05}
B.~Urgaonkar, G.~Pacifici, P.~Shenoy, M.~Spreitzer, and A.~Tantawi.
\newblock An analytical model for multi-tier internet services and its
  applications.
\newblock {\em ACM SIGMETRICS Performance Evaluation Review}, 33(1):291--302,
  2005.

\bibitem{Varnamkhasti12}
M.J. Varnamkhasti.
\newblock Overview of the algorithms for solving the multidimensional knapsack
  problems.
\newblock {\em Advanced Studies in Biology}, 4(1):37–47, 2012.

\bibitem{Vavilapalli:2013:AHY:2523616.2523633}
V.K. Vavilapalli, A.C. Murthy, C.~Douglas, S.~Agarwal, M.~Konar, R.~Evans,
  T.~Graves, J.~Lowe, H.~Shah, S.~Seth, B.~Saha, C.~Curino, O.~O'Malley,
  S.~Radia, B.~Reed, and E.~Baldeschwieler.
\newblock {Apache Hadoop YARN: Yet Another Resource Negotiator}.
\newblock In {\em Proc. ACM SOCC}, SOCC, 2013.

\bibitem{adaptive-mapper}
Rares Vernica, Andrey Balmin, Kevin~S. Beyer, and Vuk Ercegovac.
\newblock Adaptive mapreduce using situation-aware mappers.
\newblock In {\em Proceedings of the 15th International Conference on Extending
  Database Technology}, EDBT '12, pages 420--431, New York, NY, USA, 2012. ACM.

\bibitem{EuroSys17}
C.~Wang, B.~Urgaonkar, A.~Gupta, G.~Kesidis, and Q.~Liang.
\newblock Combining spot and on-demand instances for cost effective caching.
\newblock In {\em Proc. ACM EuroSys}, Belgrade, 2017.

\bibitem{sigmetrics17}
C.~Wang, B.~Urgaonkar, N.~Nasiriani, and G.~Kesidis.
\newblock {Using Burstable Instances in the Public Cloud: What, When and How?}
\newblock In {\em Proc. ACM SIGMETRICS, Urbana-Champaign, IL}, June 2017.

\bibitem{loc-based-data-partitioning}
C.~Wang, Q.~Wu, Y.~Tan, W.~Wang, and Q.~Wu.
\newblock Locality based data partitioning in mapreduce.
\newblock In {\em 2013 IEEE 16th International Conference on Computational
  Science and Engineering}, pages 1310--1317, Dec 2013.

\bibitem{BLi16b}
W.~Wang, B.~Li, B.~Liang, and J.~Li.
\newblock Multi-resource fair sharing for datacenter jobs with placement
  constraints.
\newblock In {\em Proc. Supercomputing}, Salt Lake City, Utah, 2016.

\bibitem{wondershaper}
Wondershaper.
\newblock https://github.com/magnific0/wondershaper.

\bibitem{wu2013spanstore}
Z.~Wu, M.~Butkiewicz, D.~Perkins, E.~Katz-Bassett, and H.V. Madhyastha.
\newblock Spanstore: Cost-effective geo-replicated storage spanning multiple
  cloud services.
\newblock In {\em Proc. ACM SOSP}, 2013.

\bibitem{1687544}
W.~Xu, X.~Zhu, S.~Singhal, and Z.~Wang.
\newblock Predictive control for dynamic resource allocation in enterprise data
  centers.
\newblock In {\em Proc. IEEE/IFIP NOMS}, 2006.

\bibitem{wrangler}
Neeraja~J. Yadwadkar, Ganesh Ananthanarayanan, and Randy~H. Katz.
\newblock Wrangler: Predictable and faster jobs using fewer resources.
\newblock In {\em SoCC}, 2014.

\bibitem{tr-spark}
Y.~Yan, Y.~Gao, Y.~Chen, Z.~Guo, B.~Chen, and T.~Moscibroda.
\newblock Tr-spark: Transient computing for big data analytics.
\newblock In {\em Proc ACM SoCC}, 2016.

\bibitem{yao2014online}
M.~Yao and C.~Lin.
\newblock An online mechanism for dynamic instance allocation in reserved
  instance marketplace.
\newblock In {\em Proc. IEEE ICCCN}, 2014.

\bibitem{spark-the-paper}
Matei Zaharia, Mosharaf Chowdhury, Michael~J. Franklin, Scott Shenker, and Ion
  Stoica.
\newblock Spark: Cluster computing with working sets.
\newblock In {\em Proceedings of the 2Nd USENIX Conference on Hot Topics in
  Cloud Computing}, HotCloud'10, pages 10--10, Berkeley, CA, USA, 2010. USENIX
  Association.

\bibitem{riffle}
H.~Zhang, B.~Cho, E.~Seyfe, A.~Ching, and M.J. Freedman.
\newblock {Riffle: Optimized Shuffle Service for Large-scale Data Analytics}.
\newblock In {\em Proc. EuroSys}, 2018.

\end{thebibliography}

\end{document}